\documentclass[12pt]{article}

\usepackage{amssymb,amsmath}
\usepackage{amsthm}
\usepackage{amsfonts}
\usepackage{comment}
\usepackage{indentfirst}
\usepackage{mathrsfs}
\usepackage{ascmac}
\usepackage{color}
\usepackage[top=30truemm,bottom=30truemm,left=25truemm,right=25truemm]{geometry}

\newtheorem{thm}{Theorem}[section]
\newtheorem{lemma}[thm]{Lemma}
\newtheorem{defn}[thm]{Definition}
\newtheorem{rem}[thm]{Remark}
\newtheorem{ass}[thm]{Assumption}
\newtheorem{exam}[thm]{Example}
\title{{\LARGE Spectral analysis of an abstract pair interaction model}}
\author{Keisuke ASAHARA
    \thanks{Graduate School of Science, Department of Mathematics,
               Hokkaido University\newline
              Kita 10, Nishi 8, Kita-Ku, Sapporo, Hokkaido, 060-0810, Japan \newline 
              E-mail: asahara@math.sci.hokudai.ac.jp} 
\and Daiju FUNAKAWA
   \thanks{Department of Electronics and Information Engineering, 
              Faculty of Engineering, Hokkai-Gakuen University \newline
              4-1-40, Asahi-Cho, Toyohira-Ku, Sapporo, Hokkaido, 062-8605, Japan \newline 
              E-mail: funakawa@hgu.jp}
}
\date{}

\numberwithin{equation}{section}
\renewcommand{\labelenumi}{$(\arabic{enumi})$}

\def\C{{\mathbb{C}}}
\def\N{{\mathbb{N}}}
\def\R{{\mathbb{R}}}
\def\U{\mathbb{U}}
\def\V{\mathbb{V}}
\def\W{\mathbb{W}}

\def\mU{\mathcal{U}}

\def\mM{\mathcal{M}}

\def\msD{\mathscr{D}}
\def\msF{\mathscr{F}}
\def\msH{\mathscr{H}}

\def\msL{\mathscr{L}}

\def\emps{\emptyset}

\def\nin{\notin}

\def\disp{\displaystyle}
\def\sla{\backslash}

\def\inner<#1>{\left\langle #1 \right\rangle}
\def\slim{\mathrm{s} \mbox{-}\lim}

\def\ra{\rightarrow}
\def\lra{\leftrightarrow}
\def\dra{\downarrow}

\def\rest{\upharpoonright}

\def\ga{\gamma}

\def\la{\lambda}

\def\vep{\varepsilon}

\def\om{\omega}

\def\Ub{U_{\mathrm b}}
\def\Eg{E_{\mathrm g}}
\def\Eb{E_{\mathrm b}}

\def\Ran{\mathrm{Ran}}
\def\SQ{\mathrm{d}\Gamma_{\mathrm{b}}}
\def\Segal{\Phi_{\mathrm{s}}}
\def\fb{\mathscr{F}_{\mathrm{b}}}
\def\fbfin{\mathscr{F}_{\mathrm{b,fin}}}
\def\fbz{\mathscr{F}_{\mathrm{b,0}}}

\def\Borel{\mathbf{B}^1}

\def\Ltr{L^2(\R)}
\def\Ltd{L^2(\R^d)}
\def\Dp{D_{+}}
\def\Dm{D_{-}}
\def\Dpm{D_{\pm}}
\def\Dmp{D_{\mp}}
\def\Om{\Omega}
\def\Omz{\Omega_0}
\def\Omp{\Omega_{+}}

\def\Ompm{\Omega_{\pm}}

\def\Rp{R_{+}}
\def\Rm{R_{-}}
\def\Rpm{R_{\pm}}
\def\Rmp{R_{\mp}}
\def\Tp{T_{+}}
\def\Tm{T_{-}}
\def\Tpm{T_{\pm}}
\def\Tmp{T_{\mp}}

\def\sigp{\sigma_{\mathrm{p}}}
\def\sigac{\sigma_{\mathrm{ac}}}
\def\sigsc{\sigma_{\mathrm{sc}}}

\def\Tr{\mathrm{Tr}}
\def\Ker{\mathrm{Ker}}
\def\Re{\mathrm{Re}}
\def\Im{\mathrm{Im}}
\def\lac{\la_{\mathrm{c}}}
\def\lacz{\la_{\mathrm{c},0}}
\begin{document}

\maketitle
\begin{abstract}
We consider an abstract pair-interaction model in  quantum field theory with a coupling constant $\lambda\in\R$ and
analyze  the Hamiltonian $H(\lambda)$ of the model. 
In the massive case, there exist constants $\lambda_{\rm c}<0$ and $\lambda_{{\rm c},0}<\lambda_{\rm c}$
such that, for each $\lambda \in (\lambda_{{\rm c},0},\lambda_{\rm c})\cup
(\lambda_{\rm c},\infty)$, $H(\lambda)$ is diagonalized by a proper Bogoliubov transformation, so that
the spectrum of $H(\lambda)$ is explicitly identified, where
the spectrum of $H(\lambda)$ for $\lambda>\lambda_{\rm c}$ is different from that  for $\lambda\in (\lambda_{{\rm c},0},
\lambda_{\rm c})$. As for the case $\lambda<\lambda_{{\rm c},0}$, we show that
$H(\lambda)$ is unbounded from above and below. 
In the massless case, $\lambda_{\rm c}$ coincides with $\lambda_{{\rm c},0}$.

\bigskip


\medskip


\end{abstract}


\section{Introduction}
In this paper, we consider an abstract  pair-interaction model in  quantum  field theory.
The Hamiltonian of the model is of the form
\begin{align*}
H(\lambda):=\SQ(T)+\frac{\lambda}{2}\Segal(g)^2
\end{align*}
acting in the boson Fock space $\msF_{\rm b}(\msH)$ over a Hilbert space $\msH$ (see Subsection 2.1), 
where $T$ is a self-adjoint operator on $\msH$, $\SQ(T)$ is the second quantization  operator of $T$, 
$\Segal(g)$ is the Segal field operator with test vector $g$ in $\msH$ (see Subsection 2.1) and $\lambda\in\R $ is a coupling constant. 
A model of this type is called a $\phi^2$-model.

There have been many studies on massive or massless $\phi^2$-models in concrete forms or abstract forms
(see, e.g.,  \cite{ARUn, De17, HT1, MS05, NNS16, TR15}). 
In  \cite{MS05} and \cite{TR15}, the (essential) self-adjointness of the Hamiltonian of a $\phi^2$-model 
is proved
in the case where  $\la>0$ or $|\la|$ is sufficiently small. 
In \cite{MS05}, the existence of a ground state of a $\phi^2$-model also is shown  
in the case where the quantum field under consideration is massive and $\lambda>0$.

It is well known that Hamiltonians with linear and/or  quadratic interactions in quantum fields 
may be analyzed  by the method of  Bogoliubov transformations (see, e.g., \cite{ AR81, AR89, AR1, ARUn, Be1,De17,Hi12,NNS16}). 
A typical  Bogoliubov transformation is constructed from  bounded linear operators $U,V$ and 
a conjugation operator $J$ on $\msH$ satisfying the following equations:
\begin{equation}
\label{eq:uvcondition}
\left\{
\begin{array}{rl}
U^*U-V^*V&=I,\\
U^*_JV-V^*_JU&=0,\\
UU^*-V_JV^*_J&=I,\\
UV^*-V_JU^*_J&=0,
\end{array}
\right.
\end{equation}
where $A_J:=JAJ$ and $A^*$ is the adjoint  of a densely defined linear operator $A$. 
It is well known that
there is a unitary operator $\U$ on $\msF_{\rm b}(\msH)$
which implements the Bogoliubov transformation in question
if and only if $V$ is Hilbert-Schmidt \cite{Be1,BS79,Ru78,Sh62}. 
Moreover, it is shown that, under the condition that $V$ is Hilbert-Schmidt and suitable additional conditions, 
the Hamiltonian under consideration  is unitarily equivalent via $\U$
to a second quantization operator up to a constant addition.
 For example, the Pauli-Fierz model with dipole approximation, which can be regarded as a kind of $\phi^2$-model,  
is analyzed by this method in \cite{Hi12}.

Recently, a general quadratic form Hamiltonian with a coupling constant $\lambda\in\R$ has been analyzed in \cite{NNS16} and it is shown that, in the case of a massive quantum field, 
under suitable conditions, the Hamiltonian is diagonalized by a Bogoliubov transformation.  
In \cite{De17}, the sufficient condition formulated in \cite{NNS16} to obtain the result just mentioned has been  extended. 
The spectrum of the standard pair-interaction model in physics, which is a concrete realization 
of the abstract pair-interaction model,  is  formally known \cite{HT1} 
in the case where 
$\la>\lacz$ and $\la\neq\lac$ 
for some constants $\lambda_{\rm c}$ and $\lambda_{{\rm c},0}<\lambda_{\rm c}$. 
The paper \cite{ARUn} gives a rigorous proof for that
  in the framework of the boson Fock space theory over
$\msH=L^2(\R^d)$ for any $d\in \N$ and $\la>\lac$.
 
One of the  motivations for the present work is to extend the theory developed in \cite{ARUn} with $\msH=L^2(\R^d)$ 
to the theory with $\msH$ being an abstract Hilbert  space including the case where $\lambda<\lambda_{\rm c}$.
 
It is known \cite{HT1} that  spectral properties of a pair-interaction model may depend on the range of $\lambda$ with
$\lambda_{\rm c}$ being a border point.
Hence it is important to make this aspect clear mathematically.
Therefore we analyze our model also for the region $\lambda<\lambda_{\rm c}$.
We  show that,  in the massive case with $\lambda\in (\lambda_{{\rm c},0},\lambda_{\rm c})$ also, the method of Bogoliubov transformations
can be applied to prove that
the Hamiltonian $H(\lambda)$ is unitarily equivalent to
a second quantization operator up to a constant addition.
Then we see that the spectrum of $H(\lambda)$ for $\lambda\in (\lambda_{{\rm c},0},\lambda_{\rm c})$
is different from that for $\lambda>\lambda_{\rm c}$. 
In the massless case, $\lambda_{{\rm c},0}$ coincides with $\lambda_{\rm 0}$.

The main results of the present paper include the following (1)--(3) (see Theorem 2.8 for more details):
(1) Identification of the spectra of $H(\lambda)$ for $\lambda>\lambda_{\rm c}$.
(2) Identification of the spectra of $H(\lambda)$ for 
$\lambda_{{\rm c},0}<\lambda<\lambda_{\rm c}$ (the massive case; in the massless case, 
$\lambda_{{\rm c},0}=\lambda_{\rm c}$).
In this case,
bound states different from the ground state appear.
(3) Unboundedness from above and below of $H(\lambda)$ for $\lambda<\lambda_{{\rm c},0}$.

The outline of this paper is as follows.
In  Section 2, we define our  model and recall a fundamental fact in a general theory of  Bogoliubov transformations. 
We prove the (essential) self-adjointness of $H(\lambda)$ (Theorem \ref{sa}). 
Then we state the main theorem of this paper (Theorem \ref{Main}). 
In Section 3, we construct operators $U$ and $V$ which are used to define the  Bogoliubov transformation we need.
In Section 4, we show that $U$ and $V$ satisfy \eqref{eq:uvcondition} and $V$ is Hilbert-Schmidt.
In Section 5. we prove Theorem 2.8 (1) and  calculate the ground state energy of $H(\lambda)$ in the case 
 $\lambda >\lambda _{\rm c}$. In Section 6, we prove Theorem 2.8(2).
In Section 7, we prove Theorem 2.8 (3). 
In Section 8, we consider a slightly generalized Hamiltonian of the form
$H(\eta,\lambda):=H(\lambda)+\eta\Phi_{\rm S}(f)$ with $\eta\in\R$ and $f\in\msH$.
Applying the methods and results in the preceding sections, we can 
analyze $H(\eta,\lambda)$ to identify the spectra of it. In Appendix, we state
some basic facts in the theory of boson Fock space.

\section{Preliminaries}

\subsection{The abstract Boson Fock Space}
Let $\msH$ be a Hilbert space over the complex field $\C$ with the inner product $\inner<\cdot,\cdot>_{\msH}$. The inner product is linear in the second variable and  anti-linear in the first one. The symbol $\|\cdot\|_{\msH}$ denotes the norm associated with it. We omit $\msH$ in 
$\inner<\cdot,\cdot>_{\msH}$ and $\|\cdot\|_{\msH}$, respectively if there is no danger of confusion. For each non-negative integer $n=0,1,2,\ldots, \ \otimes_{\rm s}^n\msH$ denotes the $n$-fold symmetric tensor product Hilbert space of $\msH$ with convention $\otimes_{\rm s}^0\msH :=\C$. 
Then 
\begin{equation*}
\fb(\msH):=\oplus_{n=0}^{\infty}\otimes_{\rm s}^n\msH 
\end{equation*}
is called the Boson-Fock space over $\msH$. 
For a dense subspace $\msD$ in $\msH$, $\hat{\otimes}_{\rm s}^n\msD$ denotes the algebraic $n$-fold symmetric tensor product of $\msD$ with $\hat{\otimes}_{\rm s}^0\msH :=\C$. 
Then 
\begin{equation*}
\fbfin(\msD):=\hat{\oplus}_{n=0}^{\infty}\hat{\otimes}_{\rm s}^n\msD
\end{equation*}
is a dense subspace of $\fb(\msH)$, where $\hat{\oplus}_{n=0}^{\infty}\msD_n$ denotes the  algebraic direct sum of subspace $\msD_n\subset\otimes_{\rm s}^n\msH, n=0,1,2,\ldots$. 
The finite particle vector subspace 
\begin{equation*}
\fbz(\msH):=\left\{\psi=\{\psi^{(n)}\}_{n=0}^{\infty}\in\fb(\msH)\ \middle|
\begin{array}{c}
 \psi^{(n)}\in\otimes_{\rm s}^n\msH,\ n\geq0, \rm{there\ is\ an\ integer}\ n_0\in\N\\
\rm{such \ that} \ \psi^{(n)}=0,\ \rm{for\ all}\ n\geq n_0
\end{array}
\right\}
\end{equation*}
satisfies $\fbfin(\msD)\subset\fbz(\msH)\subset\fb(\msH)$. It is dense in $\fb(\msH)$.
For a linear operator $T$ on a Hilbert space, we denote its domain by $D(T)$.

For a densely defined closable operator $T$ on $\msH$, let $T_{\rm b}^{(n)}$ be the  densely defined closed operator on $\otimes_{\rm s}^n\msH$ defined by 
\begin{equation*}
T_{\rm b}^{(n)}:=\left\{
\begin{array}{cl}
\overline{\disp\sum_{j=1}^n I\otimes\cdots\otimes I\otimes \overbrace{T}^{j-\rm{th}}\otimes I\otimes\cdots\otimes I\rest\hat{\otimes}_{\rm s}^n D(T)},
& n\geq1,\\
0,& n=0,
\end{array}
\right.
\end{equation*}
where $I$ denotes the identity operator on $\msH$, $\overline{A}$ denotes the closure of a closable operator $A$ and $A\rest \mM$ denotes the restriction of a linear operator $A$ on a subspace $\mM$. The operator 
\begin{equation*}
\SQ(T):=\oplus_{n=0}^{\infty}T_{\rm b}^{(n)}
\end{equation*}
is called the second quantization operator of $T$. If $T$ is self-adjoint or non-negative, then  so is $\SQ(T)$. 
For each $f\in\msH$, there exists a unique densely defined closed operator 
 $A(f)$ on $\fb(\msH)$ such that its adjoint $A(f)^*$ 
is given as follows:
\begin{align*}
D(A(f)^*):=\left\{\psi=\{\psi^{(n)}\}_{n=0}^{\infty}\in\fb(\msH)\ \middle|\
\sum_{n=0}^{\infty}n\left\|S_n\left(f\otimes\psi^{(n-1)}\right)\right\|^2<\infty\right\},\\
(A(f)^*\psi)^{(n)}=\sqrt{n}S_{n}(f\otimes\psi^{(n-1)}), n\in\N,\quad 
(A(f)^*\psi)^{(0)}=0\ {\rm for}\ \psi\in D(A(f)^*),
\end{align*}
where $S_n$ is the symmetrization operator on the $n$-fold tensor product $\otimes^n\msH$ of $\msH$. The operator $A(f)$ (resp. $A(f)^*$) is called the annihilation 
(resp. creation) operator with test vector $f$. 
We have 
\begin{equation*}
\fbz(\msH)\subset D(A(f))\cap D(A(f)^*)
\end{equation*}
for all $f\in\msH$ and $A(f)$ and $A(f)^*$ leave $\fbz(\msH)$ invariant. 
Moreover, they satisfy the following commutation relations:
\begin{align}
\label{eq:ccr}
[A(f),A(g)^*]&=\inner<f,g>,& [A(f),A(g)]&=0, & [A(f)^*,A(g)^*]&=0,& \mathrm{for}\ \mathrm{all}\ f,g\in\msH
\end{align}
on $\fbz(\msH)$, where $[A,B]:=AB-BA$ is the commutator of linear operators $A$ and $B$.
The relation \eqref{eq:ccr} is called the canonical commutation relations (CCR) over $\msH$. 
The symmetric operator 
\begin{equation*}
\Segal(f):=\frac{1}{\sqrt{2}}(A(f)+A(f)^*),\ f\in \msH
\end{equation*}
is called the Segal-field operator with test vector $f$. 
We write its closure by the same symbol.

\subsection{Bogoliubov Transformation}
In this subsection, we define a Bogoliubov transformation and recall an important  theorem about it. For a conjugation $J$ on $\msH$ (i.e., $J$ is an anti-linear operator on $\msH$ satisfying $\|Jf\|=\|f\|$ for all $f\in\msH$ and $J^2=I$) and a linear operator $A$ on $\msH$, we define 
\begin{equation*}
A_J:=JAJ.
\end{equation*} 
\begin{defn}
Let $U$ and $V$ be bounded linear operators on $\msH$ and $J$ be a conjugation on $\msH$. Then, for each $f\in\msH$, we define a linear operator $B(f)$ on $\fb(\msH)$ by 
\begin{equation*}
B(f):=A(Uf)+A(JVf)^*.
\end{equation*}
Then the correspondence $(A(\cdot),A(\cdot)^*)\mapsto(B(\cdot),B(\cdot)^*)$ is called a Bogoliubov transformation.
\end{defn}
By $\fbz(\msH)\subset D(B(f))$, the adjoint $B(f)^*$ exists and the equation $B(f)^*=A(Uf)^*+A(JVf)$ holds on $\fbz(\msH)$ for each $f\in\msH$. If the equations 
\begin{align*}
U^*U-V^*V=I,\quad U^*_JV-V^*_JU=0
\end{align*}
hold, then the Bogoliubov transformation preserves CCR, i.e., it holds that 
\begin{equation*}
[B(f),B(g)^*]=\inner<f,g>,[B(f),B(g)]=0,[B(f)^*,B(g)^*]=0,\ \mathrm{for}\ \mathrm{all}\ f,g\in\msH,
\end{equation*}
on $\fb(\msH)$. The following theorem is well-known \cite{Ru78, Sh62}: 
\begin{thm}
\label{PC}
Let $\msH$ be separable and $U$ and $V$ satisfy \eqref{eq:uvcondition}. 
Then there exists a unitary operator $\U$ on $\fb(\msH)$ such that 
\begin{equation*}
\U\overline{B(f)}\U^{-1}=A(f),\quad f\in\msH 
\end{equation*}
if and only if $V$ is Hilbert-Schmidt.
\end{thm}

\subsection{Hamiltonian}
\label{Hamiltonian}
For a self-adjoint operator $T$ on $\msH$, constants $\lambda,\eta\in\R$ which are called coupling constants, and vectors $f,g\in\msH$, we define an operator 
\begin{equation*}
H(\la):=\SQ(T)+\dfrac{\la}{2}\Segal(g)^2,
\quad H(\eta,\lambda):=H(\lambda)+\eta\Segal(f).
\end{equation*}
If $g\in D(T^{-1/2})$, we define the constant 
\begin{equation*}
\lambda_{{\rm c},0}:=-\|T^{-1/2}g\|^{-2}.
\end{equation*}
\begin{thm}
\label{sa}
Suppose that $T$ is an injective, non-negative, self-adjoint operator on $\msH$. 
Let $f\in D(T^{-1/2})$ and $g\in D(T^{-1/2})\cap D(T)$. Then the following $(1)$-$(3)$ hold:
\begin{enumerate}
\item
Let 
\begin{equation}
\lambda_T(g):=\|T^{-1/2}g\|^{-1}(\|T^{-1/2}g\|+\|T^{1/2}g\|)^{-1}
\end{equation}
and $|\la|<\lambda_T(g)$. 
Then $H(\eta,\la)$ is self-adjoint with $D(H(\eta,\la))=D(\SQ(T))$ and essentially self-adjoint on any core of $\SQ(T)$ for all $\eta\in\R$. Moreover $H(\eta,\la)$ is bounded from below.
\item
Let $|\la|\geq\lambda_T(g)$ and $f\in D(T^{1/2})$. 
Then $H(\eta,\la)$ is essentially self-adjoint on any core of $\SQ(T)$ for all $\eta\in\R$. 
Moreover, if $\lambda\geq\lambda_T(g)$, then $H(\eta,\lambda)$ is self-adjoint. 
\item
Let $f\in D(T^{1/2})$. Then $\overline{H(\lambda_{{\rm c},0})}$ is bounded from below. 
Moreover, if $\lambda>\lambda_{{\rm c},0}$, then $\overline{H(\eta,\lambda)}$ is bounded from below for all $\eta\in\R$ and 
$D(\SQ(T)^{1/2})=D(\overline{H(\eta,\la)}+M)^{1/2})$ for a constant $M\geq0$ satisfying $\overline{H(\eta,\la)}+M\geq0$.
\end{enumerate}
\end{thm}

\begin{proof}
\begin{enumerate}
\item
For any $\lambda\in\R$, by using \eqref{eq:ccr}, \eqref{eq:ani}, \eqref{eq:cre} and  \cite[Theorem 5.18.]{AR2}, one can easily see that there are constants $a,b\geq0$ such that for all $\psi\in D(\SQ(T))$,
\begin{equation*}
\left\|\frac{\la}{2}\Segal(g)^2\psi\right\|\leq
\dfrac{|\la|}{4}\left(a\|\SQ(T)\psi\|+b\|\psi\|\right).
\end{equation*}
In particular, we can choose $a$ and $b$ which satisfy $a|\la|/4<1$ 
if $|\la|<\lambda_T(g)$. 
We remark that, to obtain the factor $\lambda_T(g)$, we need to deform terms 
$\|A(g)^{*2}\psi\|^2$, $\|A(g)^*A(g)\psi\|^2$ and $\|A(g)^2\psi\|^2$ coming from $\|\Segal(g)\psi\|^2$
 $(\psi\in\fbz(\msH))$ to $\|A(g)A(g)^*\psi\|^2+$ a marginal term respectively. 
Thus, for $|\la|<\lambda_T(g)$, by the Kato-Rellich theorem, $H(\lambda)$ is self-adjoint. 
It is well known that $\Segal(f)$ is infinitesimally small with respect to $\SQ(T)$. 
Hence, by the Kato-Rellich theorem, for $|\la|<\lambda_T(g)$, $H(\eta,\la)$ is self-adjoint. 
\item
Firstly, we show that, for any $f\in D(T^{1/2})$ and $\eta, \lambda\in\R$, $H(\eta,\la)$ is essentially self-adjoint on any core of $\SQ(T)$. 
By \eqref{eq:ani}, \eqref{eq:cre} and \cite[Theorem 5.18.]{AR2}, we can see that there exists $a>0$ such that 
$\|H(\eta,\la)\psi\|\leq a\|(\SQ(T)+I)\psi\|$ for all $\psi\in D(\SQ(T))$. 
For the first let $f\in D(T)$. Then by \eqref{eq:ccr} and \eqref{eq:crsq}, 
for any $\psi\in\fbfin(D(T))$, we have 
\begin{align*}
&\inner<H(\eta,\la)\psi,(\SQ(T)+I)\psi>-\inner<(\SQ(T)+I)\psi,H(\eta,\la)\psi>\\
=&\frac{\la}{\sqrt2}\left(\inner<\Segal(g)\psi,A(Tg)\psi>
-\inner<A(Tg)\psi,\Segal(g)\psi>\right)
+\frac{\eta}{\sqrt2}(\inner<\psi,A(Tf)\psi>-\inner<A(Tf)\psi,\psi>).
\end{align*}
Thus, by \eqref{eq:ani} and \eqref{eq:cre}, we have 
\begin{align}
\label{eq:nelson1}
|\inner<H(\eta,\la)\psi,(\SQ(T)+I)\psi>-\inner<(\SQ(T)+I)\psi,H(\eta,\la)\psi>|
\leq C\|(\SQ(T)+I)^{1/2}\psi\|^2,
\end{align}
where $C:=\left\{|\la|\|T^{1/2}g\|(\|g\|+2\|T^{-1/2}g\|)+\sqrt{2}|\eta|\|T^{1/2}f\|\right\}$.
By a limiting argument, using the fact that $\fbfin(D(T))$ is a core of $\SQ(T)$ and $\SQ(T)$-boundedness of $\Segal(g)^2$, we can show that for 
$f\in D(T^{1/2})$ and $\psi\in D(\SQ(T))$, \eqref{eq:nelson1} holds. 
Thus, by the Nelson commutator theorem, for all $\eta,\lambda\in\R$, $H(\eta,\la)$ is essentially self-adjoint and $\overline{H(\eta,\la)}$ is essentially self-adjoint on any core of $\SQ(T)$. The equation $\overline{H(\eta,\la)\rest \msD}=\overline{\overline{H(\eta,\la)}\rest \msD}$ holds for any core $\msD$ of $\SQ(T)$. Hence $H(\eta,\la)$ is essentially self-adjoint on any core of $\SQ(T)$ for all $\eta, \lambda\in\R$. 
Next we show that, if $\lambda>-\|T^{-1/2}g\|^{-1}(\|T^{-1/2}g\|+\|T^{1/2}g\|)^{-1}$, then  $H(\eta,\lambda)$ is self-adjoint. 
We can show that , for $\lambda>0$ and any $0<\vep<1$, there is a constant $c_{\vep}>0$ such that 
\begin{equation*}
(1-\vep)\|\SQ(T)\psi\|^2+\left\|\frac{\la}{2}\Segal(g)^2\psi\right\|^2
\leq\|H(\eta,\la)\psi\|^2+c_{\vep}\|\psi\|^2, \quad \psi\in D(\SQ(T)).
\end{equation*}
Hence $H(\eta,\la)$ is closed. In particular, it is self-adjoint. 
\item
It is well known that, for any $\vep>0$, $\vep\SQ(T)+\eta\Segal(f)$ is bounded from below. 
For any $\vep>0$ and $\psi\in D(\SQ(T)^{1/2})$, 
\begin{equation*}
|\inner<\psi,A(f)\psi>|\leq \|T^{-1/2}f\|\left(\vep\|\SQ(T)^{1/2}\psi\|^2+\frac{1}{4\vep}\|\psi\|^2\right).
\end{equation*}
Hence if the assertion follows for $\eta=0$, then so is for all $\eta$. 
Thus we show that the assertion follows for $\eta=0$. 
If $\la>0$, then clearly $H(\la)\geq0$. 
Let $\la<0$. For any $\psi\in D(\SQ(T)^{1/2})$,
\begin{align*}
\|\Segal(g)\psi\|^2
\leq 2\|T^{-1/2}g\|^2\|\SQ(T)^{1/2}\psi\|^2+\|g\|^2\|\psi\|^2.
\end{align*}
Thus for any $\psi\in D(\SQ(T))$,
\begin{align}
\label{eq:bddfb}
\inner<\psi,H(\la)\psi>=&
\ \|\SQ(T)^{1/2}\psi\|^2+\dfrac{\la}{2}\|\Segal(g)\psi\|^2\nonumber\\
\geq&\ (1+\la\|T^{-1/2}g\|^2)\|\SQ(T)^{1/2}\psi\|^2+\dfrac{\la}{2}\|g\|^2\|\psi\|^2.
\end{align}
Hence $H(\la)$ is bounded from below if $\la\geq\lacz$.\\
\quad Let $\la\geq\lacz$ and $M\geq0$ be a constant satisfying $H(\la)+M\geq0$. Then for any $\psi\in D(\SQ(T))=D(H(\la)),$ 
\begin{align}
\label{eq:bddhami<}
\|(\overline{H(\la)}+M)^{1/2}\psi\|^2\leq 
(1+|\la|\|T^{-1/2}g\|^2)\|\SQ(T)^{1/2}\psi\|^2+\left(\dfrac{|\la|}{2}\|g\|^2+M\right)\|\psi\|^2.
\end{align}
By the fact that $D(\SQ(T))$ is a core of $\SQ(T)^{1/2}$, we have 
$D(\SQ(T)^{1/2})\subset D((\overline{H(\la)}+M)^{1/2})$ 
and \eqref{eq:bddhami<} holds on $D(\SQ(T)^{1/2})$.\\
\quad In the case of $\la>0$, it is easy to see that 
$\|H(\la)^{1/2}\psi\|\geq\|\SQ(T)^{1/2}\psi\|$ holds for any $\psi\in D(\SQ(T))$. 
In the case of $0>\la>\lacz$, 
\begin{equation*}
\|\SQ(T)^{1/2}\psi\|^2\leq
\dfrac{1}{1+\la\|T^{-1/2}g\|^2}\left\{\|(\overline{H(\la)}+M)^{1/2}\psi\|^2
-\left(\dfrac{\la}{2}\|g\|^2+M\right)\|\psi\|^2\right\}.
\end{equation*}
holds for any $\psi\in D(\SQ(T))$ by \eqref{eq:bddfb}.
Hence for $\la>\lacz$ there is a constant $a,b\geq0$ such that 
\begin{equation}
\label{eq:bddhami>}
\|\SQ(T)^{1/2}\psi\|\leq a\|(\overline{H(\la)}+M)^{1/2}\psi\|+b\|\psi\|.
\end{equation}
By operational calculus, $D(\SQ(T))$ is a core of $(\overline{H(\la)}+M)^{1/2}$.
Thus we have 
$D((\overline{H(\la)}+M)^{1/2})\subset D(\SQ(T)^{1/2})$ 
and \eqref{eq:bddhami>} holds on $D((\overline{H(\la)}+M)^{1/2})$.
\end{enumerate}
\end{proof}

\begin{rem}
\upshape
By \cite[Lemma13-15]{AR1}, if $\msH$ is separable, then Theorem \ref{sa} takes the following forms:\\
Let $\msH$ be separable, $T$ be a non-negative, injective self-adjoint operator, $f\in D(T^{-1/2})$ and $g\in D(T^{-1/2})\cap D(T^{1/2})$. Then the following (1)-(3) holds:
\begin{enumerate}
\item
Let $\la>\lacz$. 
Then $H(\eta,\la)$ is self-adjoint with $D(H(\eta,\la))=D(\SQ(T))$ and essentially self-adjoint on any core of $\SQ(T)$ for all $\eta\in\R$. Moreover $H(\eta,\la)$ is bounded from below.
\item
Let $\la\leq\lacz$ and $f\in D(T^{1/2})$. 
Then $H(\eta,\la)$ is essentially self-adjoint on any core of $\SQ(T)$ for all $\eta\in\R$. 
In particular, if $\eta=0$ and $\la=\lacz$, then $H(\lacz)=H(0,\lacz)$ is bounded from below. 
\item
Let $\la>\lacz$. Then $D(\SQ(T)^{1/2})=D(H(\eta,\la)+M)^{1/2})$ for a constant $M\geq0$ satisfying $H(\eta,\la)+M\geq0$.
\end{enumerate}
\end{rem}

\begin{defn}
Let $T$ be a self-adjoint operator on $\msH$ and $\{E(B)\ |\ B\in\Borel\}$ be the 
spectral measure associated with $T$ on the Borel field $\Borel$ on $\R$. 
The operator $T$ is called purely absolutely continuous if, for each $f\in\msH$, the measure $\|E(\cdot)f\|^2$ on $\Borel$ is absolutely continuous with respect to the one-dimensional  Lebesgue measure.
\end{defn}

\begin{defn}
For a purely absolutely continuous self-adjoint operator $T$ and vectors $f,g\in\msH$, $\psi_{g,f}$ denotes the Radon-Nikodym derivative of the finite complex Borel measure  $\inner<g,E(\cdot)f>$ on $\Borel$. In particular, we set $\psi_g:=\psi_{g,g}$. 
\end{defn}

\subsection{Assumptions}
To prove our main theorem stated later (Theorem \ref{Main}), we need some assumptions. 
For a closed operator $A$, $\sigma(A)$ denotes the spectrum of $A$. If $A$ is self-adjont,  then $\sigma_{\rm ac}(A)$ (resp. $\sigma_{\rm p}(A)$, $\sigma_{\rm sc}(A)$) 
denotes the absolutely continuous (resp. point, singular continuous) spectrum of $A$. 
For a self-adjoint operator $A$ bounded from below, 
\begin{equation*}
E_0(A):=\inf\sigma(A) 
\end{equation*}
is called the lowest energy of $A$. In particular, it is called the ground state energy of $A$ if $E_0(A)\in\sigp(A)$. In this case, any for responding eigenvector is called a ground state of $A$. The ground state is said to be unique if ${\rm dim}\ \Ker(A-E_0(A))=1$. For linear operators  $A$ and $B$, the symbol $A\subset B$ means that $D(A)\subset D(B)$ and $Af=Bf$ for all  $f\in D(A)$, i.e., $B$ is an extension of $A$.

\begin{ass}
\label{ass}
\begin{enumerate}
\item
The operator $T$ is a non-negative, purely absolutely continuous self-adjoint operator,
\item
The fixed vector $g\in\msH$ satisfies 
$g\in D(\hat{T}^{-1/2})\cap D(T^{1/2})$ and $Jg=g$, where $\hat{T}:=T-E_0$, $E_0:=E_0(T)$ and 
 $J$ is a conjugation on $\msH$ satisfying 
$J D(T)\subset D(T)$ and $JT\psi=TJ\psi$ for any $\psi \in D(T)$ $($ i.e., $JT\subset TJ)$,
\item
$\sup_{E_0<x}x^{\pm1}\psi_g(x)<\infty$ and $\psi_g(x)>0$ for all  $x\in(E_0,\infty)$, 
\item
$\psi_g\in C^1((E_0,\infty))$ and  
$\lim_{x\dra E_0}x^{-1}\psi_g'(x)=0=\lim_{x\ra\infty}x^{-1}\psi_g'(x).$
\end{enumerate}
\end{ass}

\begin{rem}\upshape 
\label{remass}
The operator $T$ is injective since it is a purely absolutely continuous self-adjoint operator. 
Since $T$ has no eigenvector, the inverse of $\hat{T}$ exists. 
Assumption \ref{ass} (2) implies that $T_J=T$. In general, for a self-adjoint operator $A$ and a conjugation $J$, we can choose a vector $f\in D(A)$ satisfying $Jf=f$ if $A_J=A$. Thus the vector $g$ in Assumption \ref{ass} (2) exists. 
By Assumption \ref{ass} (3), one can easily show that $\sup_{x\in\sigma(T)}\psi_g(x)<\infty$ 
and, for each $f\in\msH$, the functions $\psi_{g,f},\psi_{T^{\pm1/2}g,f}$ are in $\Ltr$ and the maps $: f\mapsto\psi_{g,f}, \psi_{T^{\pm1/2}g,f}$ are bounded. Actually, for any $h\in\msH$ and $B\in\Borel$, the following inequality holds
\begin{equation*}
|\inner<E(B)h,f>|^2\leq\|E(B)h\|^2\|E(B)f\|^2
\end{equation*}
by Schwarz's inequality. 
Thus we obtain $|\psi_{h,f}(\mu)|^2\leq\psi_h(\mu)\psi_f(\mu)$ for almost all $\mu\in\R$ with respect to the Lebesgue measure. Hence, by Assumption \ref{ass} (3), 
we have the boundedness of the mappings. 
Moreover, we see that for any $F\in\Ltr$, $g\in D(F(T))$, where $F(T)$ denotes the operator  defined by $F(T):=\int_{\R}F(\mu)dE(\mu)$. 
In particular, $g$ is in $D(\psi_{g,f}(T))$ for any $f\in\msH$. 
\end{rem}

\begin{lemma}
\label{f(T)_J=f(T)}
Let  $T$ be a self-adjoint operator such that $JT\subset TJ$. 
Then 
\begin{enumerate}
\item
$E(B)_J=E(B),$ for all $B\in\Borel$.
\item
Let $F$ be a Borel measurable function on $\R$. Then $F(T)_J=F^*(T),$ 
where $F^*$ is complex conjugation of $F$.
\end{enumerate}
\end{lemma}

\begin{proof}
These are proved by using the spectral theorem.
\end{proof}

\subsection{The Main Theorem}
In this subsection, we state the main theorem of the present paper. 
Let $\lambda_{\rm c}$ be a constant defined by 
\begin{equation*}
\lac:=-\left(\int_{[E_0,\infty)}\frac{\mu}{\mu^2-E_0^2}\ d\|E(\mu)g\|^2\right)^{-1}<0.
\end{equation*}
Then it is easy to see that $\lacz\leq\lac$, and $\lacz=\lac$ if and only if $E_0=0$.
\begin{thm}
\label{Main}
Let $\msH$ be separable. Then the following $(1)$-$(3)$ hold:
\begin{enumerate}
\item
Let $T$ and $g$ satisfy Assumption \ref{ass}. If $\la>\lac$, then there are a unitary operator $\U$ on $\fb(\msH)$ and a constant $\Eg\in\R$ such that 
\begin{equation}
\label{eq:main1}
\U H(\la)\U^{-1}=\SQ(T)+\Eg.
\end{equation}
In particular,  $\U^{-1}\Omz$ is the unique ground state of $H(\la)$ up to constant multiples,  and  
\begin{align}
\sigma(H(\la))&=\{\Eg\}\cup[E_0+\Eg,\infty),\label{eq:spectrum}\\
\sigac(H(\la))&=[E_0+\Eg,\infty),\ \sigp(H(\la))=\{\Eg\},\ \sigsc(H(\la))=\emps.\label{eq:spectra}
\end{align}
\item
Let $T$ and $g$ satisfy Assumption \ref{ass} and $E_0>0$. 
If $\lacz<\la<\lac$, then there exist a unitary operator $\V$ on $\fb(\msH)$, an injective non-negative self-adjoint operator $\xi$ on $\msH$ and a constant $\Eb\geq0$ such that 
$\xi$ has a ground state and 
\begin{equation*}
\V H(\la)\V^{-1}=\SQ(\xi)+\Eg-\Eb.
\end{equation*}
In particular, $\V^{-1}\Omz$ is the unique ground state of $H(\lambda)$ up to constant multiples, and  
\begin{align*}
&\sigma(H(\la))=\cup_{n=0}^{\infty}\{n\beta+\Eg-\Eb\}\cup[E_0+\Eg-\Eb,\infty),\\
&\sigac(H(\la))=[E_0+\Eg-\Eb,\infty),\\ &\sigp(H(\la))=\cup_{n=0}^{\infty}\{n\beta+\Eg-\Eb\},\ 
\sigsc(H(\la))=\emps,
\end{align*}
where $\beta>0$ is the discrete ground state energy of $\xi$.
\item
Let $T$ be a non-negative, injective self-adjoint operator. 
If $g\in D(T^{-1/2})$ and $\la<\lacz$, then $H(\la)$ is unbounded from above and below.
\end{enumerate}
\end{thm}

\begin{exam}
\upshape
A concrete realization of the abstract model is given  as follows (see \cite[Chapter 12]{HT1}): 
\begin{equation*}
\begin{array}{ccc}
\msH\lra \Ltd,&T\lra\om,&g\lra\dfrac{\hat{\rho}}{\sqrt{\om}}
\end{array}
\end{equation*}
where $\om$ is a multiplication operator associated with the function $\om(k):=\sqrt{|k|^2+m^2},k\in\R^d$ for a fixed $m\geq0$ and $\hat{\rho}$ is the Fourier transform of a function $\rho\in\Ltd$ satisfying $\hat{\rho}/\sqrt{\om}\in\Ltd$.
Assume that $\hat{\rho}$ is rotation invariant, i.e., there exists a function $v$ on $[0,\infty)$ such that 
$\hat{\rho}(k)=v(|k|)$ for all $k\in\R^d$. Then we have $\psi_g(s)=|S^{d-1}|\ \om_1^{-1}(s)^{d-2}\ |v(\om_1^{-1}(s))|^2$ for $s\geq m$, where $|S^{d-1}|$ is the surface area of the $(d-1)$-dimensional unite sphere with convention $|S^0|=2\pi$ and $\om_1(r)=\sqrt{r^2+m^2},r\geq0$. 
Hence, with $J$ being the complex conjugation, the following conditions 
(2)'-(4)' imply that the present model satisfies Assumption \ref{ass}:
\begin{enumerate}
\renewcommand{\labelenumi}{(\arabic{enumi})'}
\setcounter{enumi}{1}
\item
$\hat{\rho}(k)^*=\hat{\rho}(k)$ and 
\begin{equation*}
\hat{\rho}\in\Ltd,\int_{\R^d}\dfrac{|\hat{\rho}(k)|^2}{|k|^2}dk<\infty.
\end{equation*}
\item
$\hat{\rho}$ is rotation invariant.
$\sup_{k\in\R^d}\om(k)^{\pm1/2}|k|^{(d-2)/2}|\hat{\rho}(k)|<\infty$.
$\hat{\rho}(k)>0$, for all $k\neq0$.
\item
$v\in C^1([0,\infty))$ and 
\begin{align*}
\lim_{|k|\ra0}|k|^{d-4}\hat{\rho}(k)\{(d-2)\hat{\rho}(k)+2v'(|k|)\}=0.\\
\lim_{|k|\ra\infty}|k|^{d-4}\hat{\rho}(k)\{(d-2)\hat{\rho}(k)+2v'(|k|)\}=0.
\end{align*}
\end{enumerate}
For example, one can easily check that the function
\begin{equation*}
\hat{\rho}(k):=\exp{\left(-\frac{1}{|k|^2}-|k|^2\right)},\ k\in\R^d\sla\{0\},\ \hat{\rho}(0):=0
\end{equation*}
satisfies the above conditions (2)'-(4)'.
\end{exam}

\section{Definitions and properties of some functions and operators}
In this section we introduce some functions and operators. 
We assume that $\msH$ is separable and Assumption \ref{ass} from this section to Section 6.
\subsection{Functions $D$ and $\Dpm$}
\begin{lemma}
\label{proD}
Let $D:\C\backslash(0,\infty)\ra\C$ be the function 
\begin{equation*}
D(z):=1+\la\int_{[E_0,\infty)}\dfrac{\mu}{\mu^2-E_0^2-z}d\|E(\mu)g\|^2
,\quad z\in\C\sla(0,\infty).\\
\end{equation*}
Then $D$ is well-defined and analytic in $\C\backslash[0,\infty)$. 
Moreover, the following hold$:$
\begin{enumerate}
\item
For all $\la>\la_c$, $D(z)$ has no zeros in $\C\backslash[0,\infty).$ 
\item
For all $\la<\la_c$, $D(z)$ has a unique simple zero in the negative real axis $(-\infty,0).$
\end{enumerate}
\end{lemma}

\begin{proof}
If $\Im z\neq0$ (resp. $\Re z<0$), then for any $n\in\N$, 
\begin{equation*}
\int_{[E_0,\infty)}\left|\dfrac{\mu}{(\mu^2-E_0^2-z)^n}\right|d\|E(\mu)g\|^2
\leq c^{-n}\|T^{1/2}g\|^2<\infty,
\end{equation*}
where $c$ is $|\Re z|$ (resp. $|\Im z|$).
If $z=0$, then 
\begin{equation*}
\int_{[E_0,\infty)}\dfrac{\mu}{\mu^2-E_0^2}d\|E(\mu)g\|^2
\leq\|\hat{T}^{-1/2}g\|^2<\infty.
\end{equation*}
Thus, by using the Lebesgue 
dominated convergence theorem, $D$ is well-defined and analytic in $\C\backslash[0,\infty)$. 
\begin{enumerate}
\item
If $\la=0,$ then $D(z)=1$ for all $z\in\C\sla(0,\infty)$, so it has no zeros. 
Let $\la\neq0$ and $z=x+iy\in\C\sla(0,\infty)$. Then we see that  
\begin{equation*}
\Im\ D(z)=y\la\int_{[E_0,\infty)}\dfrac{\mu}{(\mu^2-E_0^2-x)^2+y^2}d\|E(\mu)g\|^2.
\end{equation*}
Thus $\Im\ D(z)=0$ is equivalent to $y=0.$ Therefore $D(z)=0$ if and only if $D(x)=0.$
Let $y=0.$ In the case $\la>0$, one has $D(x)>0$ for all $x\in(-\infty,0]$. Thus $D$ has no zeros. Next, we consider the case $\la<0$. We have for $x<0$,
\begin{equation*}
D'(x)=\la\int_{[E_0,\infty)}\dfrac{\mu}{(\mu^2-E_0^2-x)^2}d\|E(\mu)g\|^2<0.
\end{equation*}
Thus $D$ is monotone decreasing in $(-\infty,0)$. If $\la>\la_c$, then $D(0)>0$. 
Hence $D$ has no zeros. 
\item
Let $\la<\la_c$. We can see that 
\begin{equation*}
D(0)=1+\la\int_{[E_0,\infty)}\dfrac{\mu}{\mu^2-E_0^2}d\|E(\mu)g\|^2=1-\dfrac{\la}{\la_c}<0.
\end{equation*}
By the Lebesgue dominated convergence theorem $D(x)\ra1$ as $x\ra-\infty$. 
Since $D$ is monotone decreasing in $(-\infty,0)$, $D$ has a unique simple zero in $(-\infty,0)$.
\end{enumerate}
\end{proof}

Let 
\begin{equation*}
\phi_g(x):=\psi_g(\sqrt{x})\chi_{[E_0^2,\infty)}(x), x\in\R,
\end{equation*}
where $\chi_B$ is the characteristic function of $B\in\Borel$. 

\begin{lemma}
\label{hiltraconv}
The following hold$:$
\begin{enumerate}
\item
The function $\phi_g$ satisfies 
$\phi_g\in C^1(\R)\cap L^1(\R)\cap L^2(\R)$ 
and $\sup_{x\in\R}|\phi'_g(x)|<\infty$.
\item
Let 
\begin{align*}
A_{\vep}^{(1)}(x):=\frac{x}{\pi(x^2+\vep^2)},\ 
A_{\vep}^{(2)}(x):=\frac{\vep}{\pi(x^2+\vep^2)},\ x\in\R,\ \vep>0
\end{align*}
be the conjugate poisson kernel and the poisson kernel respectively and  $f*h$ denote the convolution of  functions $f$ and $h$. Let 
\begin{equation*}
(H_{\vep}f)(s):=\frac{1}{\pi}\int_{|x-s|\geq\vep}\frac{f(x)}{x-s}dx,\ 
(Hf)(s):=\lim_{\vep\dra0}(H_{\vep}f)(s),\   s\in\R,\ \vep>0, 
\end{equation*}
where $Hf$ is called the Hilbert transform of $f$. 
Then for all $x\in\R$, 
\begin{equation*}
\lim_{\vep\dra0}\left(A_{\vep}^{(1)}*\phi_g\right)(x)=(H\phi_g)(x), \quad 
\lim_{\vep\dra0}\left(A_{\vep}^{(2)}*\phi_g\right)(x)=\phi_g(x),
\end{equation*}
hold uniformly in $x$.
\end{enumerate}
\end{lemma}

\begin{proof}
By Assumption \ref{ass} (2), (3) and (4), the assertion (1) holds. 
Next we consider the assertion (2). 
By (1), in particular, $\phi_g$ is bounded and uniformly continuous. 
Thus it is easy to see that $A_{\vep}^{(2)}*\phi_g$ converges uniformly to $\phi_g$. 
Moreover, by (1), Holder's inequality, the mean value theorem and a similar estimate to  the proof of \cite[Theorem 92.]{Tit1}, we can show that 
 $(A_{\vep}^{(1)}*\phi_g)(x)-(H_{\vep}\phi_g)(x)$ tends to $0$ uniformly in $x$ as $\vep\dra0$. 
Hence the assertion (2) holds.
\end{proof}

Detailed studies of the Hilbert transform are given in \cite{Tit1}. 

\begin{lemma}
\label{Dpm}
For all $s\geq0$, 
$\Dpm(s):=\lim_{\vep\dra0}D(s\pm i\vep)$ are uniformly convergent and continuous in $s\geq0$ with  
\begin{equation}
\label{eq:dpm}
\Dpm(s)=1+\dfrac{\la\pi}{2}(H\phi_g)(E_0^2+s)
 \pm i\dfrac{\la\pi}{2}\psi_g\left(\sqrt{E_0^2+s}\right),\quad s\geq0.
\end{equation}
\end{lemma}

\begin{proof}
For any $s\geq0$ and $\vep>0$, we have by change of variable
\begin{align*}
D(s\pm i\vep)=
\dfrac{\la\pi}{2}\left(A_{\vep}^{(1)}*\phi_g\right)(E_0^2+s)
\pm i\dfrac{\la\pi}{2}\left(A_{\vep}^{(2)}*\phi_g\right)(E_0^2+s).
\end{align*}
Thus, by Lemma \ref{hiltraconv}, $\Dpm$ converge uniformly in $s\geq0$ and \eqref{eq:dpm} holds. The continuity of $\Dpm$ is due to the uniform convergence. 
\end{proof}

\begin{rem}
For all $\mu\in[E_0,\infty)$, we have 
\begin{equation}
\label{eq:psigpm}
i\pi\la\psi_g(\mu)=\Dp(\mu^2-E_0^2)-\Dm(\mu^2-E_0^2).
\end{equation}
\end{rem}

\begin{lemma}
\label{Dpm>0}
Let $\la\neq\la_c$, then $\delta:=\inf_{s\geq0}|\Dpm(s)|>0$.
\end{lemma}

\begin{proof}
If $\la=0$, then clearly $\Dpm(s)=1>0$ for all $s\in[0,\infty)$. 
Let $\la\neq0,\lac$. Then $\Dpm(0)=D(0)\neq0$. 
Hence, by the continuity of $\Dpm$, $\Dpm$ has no zeros near $s=0$. 
By the property that 
 $\phi'_g(x)\ra0$ as $x\ra\infty$ and some estimate of $H\phi_g$, we can see that $(H\phi_g)(x)\ra0$ as $x\ra\infty$. 
This fact implies that $\inf_{s_0\leq s}\Re\Dpm(s)>0$ for a sufficiently large number $s_0>0$. In addition, $\Im\Dpm(s)$ are positive for any closed interval included in $(0,\infty)$ by Assumption \ref{ass} (3) and the continuity of $\psi_g$. Hence we can see that $\inf_{s\geq0}|\Dpm(s)|>0$.
\end{proof}

\begin{rem}
By Lemmas \ref{Dpm} and \ref{Dpm>0}, we can see that there are constants $c,d,\vep_0>0$  with $0<c<d$ and $\vep_0>0$ such that 
\begin{equation}
\label{eq:dspmbdd}
c\leq\left|\frac{D(s\pm i\vep)}{\Dpm(s)}\right|\leq d
\end{equation}
for all $s\geq0, 0<\vep<\vep_0$.
\end{rem}
\subsection{Operators $\Rpm$}
\begin{lemma}
\label{Rbdd}
One can define bounded operators $\Rpm$ on $\msH$ as follows:
\begin{equation*}
\Rpm f:=-\la\lim_{\vep\dra0}\int_{[E_0,\infty)}
\dfrac{R_{\mu'^2\pm i\vep}(T^2)T^{1/2}g}
{\Dpm(\mu'^2-E_0^2)}d\inner<T^{1/2}g,E(\mu')f>,\quad f\in\msH,
\end{equation*}
where $R_z(A)$ is the resolvent of a linear operator $A$ at $z\in\rho(A)\ ($the  resolvent set of a linear operator $A )$.
\end{lemma}

\begin{proof}
For a fixed $\vep>0$ and any $f\in\msH$,
\begin{equation*}
\int_{[E_0,\infty)}
\left\|\dfrac{R_{\mu'^2\pm i\vep}(T^2)T^{1/2}g}{\Dpm(\mu'^2-E_0^2)}\right\|d\|E(\mu')f\|^2
\leq\dfrac{\|f\|^2\|T^{1/2}g\|}{\delta\vep}<\infty
\end{equation*}
by Lemma \ref{Dpm>0} and a property of a resolvent. Thus we can define linear operators $R_{\pm}^{(\vep)}$ on $\msH$ by  
\begin{equation*}
\Rpm^{(\vep)}f:=-\la\int_{[E_0,\infty)}
\dfrac{R_{\mu'^2\pm i\vep}(T^2)T^{1/2}g}{\Dpm(\mu'^2-E_0^2)}d\inner<T^{1/2}g,E(\mu')f>
\end{equation*}
in the sense of Bochner integral with the polarization identity. For any $h,f\in\msH$,
\begin{align*}
&\inner<h,R_{\pm}^{(\vep)}f>\nonumber\\
=&-\la\int_{[E_0,\infty)}
\dfrac{\inner<h,R_{\mu'^2\pm i\vep}(T^2)T^{1/2}g>}
{\Dpm(\mu'^2-E_0^2)}d\inner<T^{1/2}g,E(\mu')f>\\
=&-\la\int_{[E_0,\infty)}\int_{[E_0,\infty)}
\dfrac{\mu^{1/2}}{(\mu^2-\mu'^2\mp i\vep)D_{\pm}(\mu'^2-E_0^2)}
d\inner<h,E(\mu)g>d\inner<T^{1/2}g,E(\mu')f>,
\end{align*}
where we have used the functional calculus. By change of variables in the Lebesgue-Stieltjes integration, functional calculus and Fubini's theorem, we have 
\begin{align*}
\inner<h,\Rpm^{(\vep)}f>=\dfrac{\la\pi}{2}\int_{[E_0,\infty)}
\left(A_{\vep}^{(1)}*\phi_{g,f}^{\pm}\right)(\mu^2)\mu^{1/2}
\mp i\left(A_{\vep}^{(2)}*\phi_{g,f}^{\pm}\right)(\mu^2)\mu^{1/2}d\inner<h,E(\mu)g>,
\end{align*}
where $\phi_{g,f}^{\pm}(x)=\psi_{g,f}(\sqrt{x})x^{-1/4}D_{\pm}(x-E_0^2)^{-1}\chi_{[E_0^2,\infty)}(x), x\in\R$. 
We have $\phi_{g,f}^{\pm}\in\Ltr$ by Remark \ref{remass}, and the function  $\left(A_{\vep}^{(j)}*\phi_{g,f}^{\pm}\right)(\mu^2)\mu^{1/2}\ (\mu\in\R)$ is in $\Ltr$ for each $j=1,2$. 
Thus we have 
\begin{equation*}
\Rpm^{(\vep)}f\ra (\pi\la/2)(H\phi_{g,f}^{\pm})(T^2)T^{1/2}g\mp (1/2)A_{\pm}f\ 
\text{as}\ \vep\dra0
\end{equation*}
by a property of Hilbert transform and the continuity of the inner product with $\Ltr$, where the linear operators 
\begin{equation*}
A_{\pm}f:= i\pi\la\psi_{g,f}(T)D_{\pm}(T^2-E_0^2)^{-1}g, f\in\msH
\end{equation*}
are well-defined (see Remark \ref{remass} and Lemma \ref{Dpm>0}). 
Moreover, by change of variables, the isometry of Hilbert transform and Remark \ref{remass},  we can show that the inequalities 
\begin{equation*}
\label{eq:bddhapm}
\left\|(H\phi^{\pm}_{g,f})(T^2)T^{1/2}g\right\|\leq\frac{c_g}{\delta}\|f\|,\quad 
\|A_{\pm}f\|\leq\frac{2\pi|\la|c_g}{\delta}\|f\|
\end{equation*}
hold for all $f\in\msH$ with constant $c_g:=\sup_{\sigma(T)}\psi_g$. 
Hence $\Rpm$ are bounded.
\end{proof}

It is easy to see that $\Rpm^*:=(\Rpm)^*$ are given as follows: for $f\in\msH$,
\begin{align}
\Rpm^{(\vep)*}f&=\la\int_{[E_0,\infty)}R_{\mu'^2\pm i\vep}(T^2)
D_{\mp}(T^2-E_0^2)^{-1}T^{1/2}g \ d\inner<T^{1/2}g,E(\mu')f>.\nonumber\\
\Rpm^*f&=\lim_{\vep\dra0}\Rpm^{(\vep)*}f.\label{eq:rpm*}
\end{align}
For a densely defined linear operator $A$ on a Hilbert space, we denote by $A^{\sharp}$ $A$ or $A^*$.

\begin{lemma}
\label{ranrpm}
The ranges of $\Rpm^{\sharp}$ are included in $D(T^{-1})\cap D(T)$ and $\Rpm^{\sharp}$ map $D(T)$ into $D(T^2)$.
\end{lemma}

\begin{proof}
For any $f,h\in\msH$, we have 
\begin{align}
\inner<h,R_{\pm}f>
=&\dfrac{\la\pi}{2}\int_{[E_0,\infty)}
\left(H\phi_{g,f}^{\pm}\right)(\mu^2)\mu^{1/2}
\mp i\dfrac{\psi_{g,f}(\mu)}{D_{\pm}(\mu^2-E_0^2)}d\inner<h,E(\mu)g>.
\label{eq:rpm}
\end{align}
By change variable, we have  
\begin{align}
(H\phi_{g,f}^{\pm})(\mu^2)
=\left(H\psi_{T^{-1/2}g,f}^{\pm}\right)(\mu)+\left(H\psi_{T^{-1/2}g,f}^{\pm}\right)(-\mu), 
 \mu\in\R,
\label{eq:hil}
\end{align}
where $\psi_{h,f}^{\pm}(x):=\psi_{h,f}(x)\Dpm(x^2-E_0)^{-1}\chi_{[E_0,\infty)}(x), x\in\R$ for $h,f\in\msH$. 
Thus we see by Assumption \ref{ass} (3) and functional calculus that $\Ran(\Rpm)\subset D(T^{-1})$. 
The equation
\begin{equation}
\label{eq:muhil}
\mu\left(H\phi_{g,f}^{\pm}\right)(\mu^2)
=\left(H\psi_{T^{1/2}g,f}^{\pm}\right)(\mu)-\left(H\psi_{T^{1/2}g,f}^{\pm}\right)(-\mu),\ \mu\in\R,
\end{equation}
operational calculus for \eqref{eq:rpm} and Assumption \ref{ass} (3) imply that  $\Ran(\Rpm)\subset D(T)$. For any $f\in D(T)$ and $\mu\in\R$,
\begin{equation*}
\label{eq:mu^2hil}
\mu^2\left(H\phi_{g,f}^{\pm}\right)(\mu^2)=
\left(H\psi_{T^{1/2}g,Tf}^{\pm}\right)(\mu)+\left(H\psi_{T^{1/2}g,Tf}^{\pm}\right)(-\mu)
-\dfrac{2}{\pi}\int_{[E_0,\infty)}\psi_{T^{1/2}g,f}^{\pm}(x)\ dx.
\end{equation*}
Hence $\Rpm f\in D(T^2)$ and the following equation holds for any $h\in\msH$,	
\begin{align*}
\inner<h,T^2\Rpm f>=&\dfrac{\la\pi}{2}\int_{[E_0,\infty)}
\left\{\left(H\psi_{T^{1/2}g,Tf}^{\pm}\right)(\mu)
+\left(H\psi_{T^{1/2}g,Tf}^{\pm}\right)(-\mu)-\frac{2c}{\pi}\right\}\mu^{1/2}
d\inner<h,E(\mu)g>\nonumber\\
&\mp i\dfrac{\la\pi}{2}\int_{[E_0,\infty)}
\psi_{T^{1/2}g,Tf}^{\pm}(\mu)\mu^{1/2}\ d\inner<h,E(\mu)g>,
\end{align*}
where $c:=\int_{\R}\psi_{T^{1/2}g,f}^{\pm}(x)dx$. 
In quite the same manner as in the case of $\Rpm$, we can prove the statement for $\Rpm^*$. 
\end{proof}

\begin{lemma}
The operator equations $(\Rpm)_J=\Rmp$ hold.
\end{lemma}

\begin{proof}
This follows from Assumption \ref{ass} (1) and Theorem \ref{f(T)_J=f(T)}.
\end{proof}

\begin{lemma}
The operator equation $\Rm=\Rp\ga+A_{-}$ holds,
where 
\begin{equation*}
\ga:=\Dp(T^2-E_0^2)D_{-}(T^2-E_0^2)^{-1}
\end{equation*}
is a bounded operator.
\end{lemma}

\begin{proof}
The first resolvent formula gives that, for any $\mu',\mu''\in\R, \vep>0$,
\begin{equation*}
R_{\mu'^2-i\vep}(T^2)-R_{\mu'^2+i\vep}(T^2)
=-2i\vep R_{\mu'^2-i\vep}(T^2)R_{\mu'^2+i\vep}(T^2).
\end{equation*}
Then, for any $f\in\msH$,
\begin{align*}
\Rm^{(\vep)}f=&-\la\int_{[E_0,\infty)}
\dfrac{R_{\mu'^2+i\vep}(T^2)T^{1/2}g}{\Dm(\mu'^2-E_0^2)}d\inner<T^{1/2}g,E(\mu')f>\\
&+2i\la\vep\int_{[E_0,\infty)}
\dfrac{R_{\mu'^2+i\vep}(T^2)R_{\mu'^2-i\vep}(T^2)T^{1/2}g}
{\Dm(\mu'^2-E_0^2)}d\inner<T^{1/2}g,E(\mu')f>.
\end{align*}
Thus, by change of variable, we have for any $h\in\msH$
\begin{align*}
\inner<h,\Rm^{(\vep)}f>=&\inner<h,\Rp^{(\vep)}\ga f>
+2i\la\int_{[E_0,\infty)}\int_{[E_0,\infty)}d\inner<h,E(\mu)g>d\inner<T^{1/2}g,E(\mu')f>\\
&\times
\dfrac{\mu^{1/2}\vep}{\{(\mu^2-\mu'^2)^2+\vep^2\}\Dm(\mu'^2-E_0^2)}\\
=&\inner<h,R_{+}^{(\vep)}\ga f>
+i\pi\la\int_{[E_0,\infty)}\left(A_{\vep}^{(2)}*\phi_{g,f}^{-}\right)(\mu^2)\mu^{1/2}d\inner<h,E(\mu)g>.
\end{align*}
By a property of the Poisson kernel, the function $\left(A_{\vep}^{(2)}*\phi_{g,f}^{-}\right)(\mu^2)\mu^{1/2}\ (\mu\in\R)$ converges to $\psi_{g,f}(\mu)/\Dm(\mu^2-E_0^2)$ as $\vep\ra+ 0$ in the sense of $L^2(\R)$. Hence the continuity of inner product with $\Ltr$ implies that
\begin{align*}
\inner<h,\Rm f>
&=\inner<h,\Rp\ga f>+i\pi\la\int_{[E_0,\infty)}
\dfrac{\psi_{g,f}(\mu)}{\Dm(\mu^2-E_0^2)}d\inner<h,E(\mu)g>\\
&=\inner<h,\Rp\ga f>+\inner<h,A_{-}f>.
\end{align*}
Since $f$ and $h$ are arbitrary, one obtains the conclusion.
\end{proof}
It is easy to see that 
\begin{equation*}
(A_{-})^*=-A_{+}.
\end{equation*}

\begin{lemma}
\label{AF(T)}
For any Borel measurable function $F$ on $\R$, 
$A_{\pm}F(T)\subset F(T)A_{\pm}$. 
\end{lemma}

\begin{proof}
It is easy to see that for any $f\in D(F(T))$, $\psi_{g,F(T)f}=F\psi_{g,f}\in L^2(\R)$. 
This fact and Lemma \ref{Dpm>0} imply that $\psi_{g,f}(T)g\in D(F(T))$ and $F(T)\psi_{g,f}(T)g=\psi_{g,F(T)f}(T)g$. Hence 
$A_{\pm}f\in D(F(T))$ and $F(T)A_{\pm}f=A_{\pm}F(T)f$.
\end{proof}

\begin{lemma}
\label{ARpm}
The following operator equations hold:
\begin{equation*}
A_{-}\Rpm^*=(\ga-I)\Rpm^*,\quad A_{-}(A_{-})^*=-A_{-}-(A_{-})^*.
\end{equation*}
\end{lemma}

\begin{proof}
By applying Lemma \ref{AF(T)} to the case $F=\chi_B$, one can easily see that $A_{\pm}E(B)=E(B)A_{\pm}$ holds for any $B\in\Borel$. 
For any $f,h\in\msH$, we have 
\begin{align*}
&\inner<(A_{-})^*h,\Rpm^{(\vep)*}f>\\
=&\ i\pi\la^2\int_{[E_0,\infty)}\int_{[E_0,\infty)}
\dfrac{\mu^{1/2}\psi_g(\mu)}
{(\mu^2-\mu'^2\mp i\vep)\Dmp(\mu^2-E_0^2)\Dm(\mu^2-E_0^2)}d\inner<h,E(\mu)g>d\inner<T^{1/2}g,E(\mu')f>.
\end{align*}
Then, since $\ga$ and $E(B)$ commute on $\msH$ for any $B\in\Borel$, 
\eqref{eq:psigpm} gives 
\begin{align*}
&\ \inner<(A_{-})^*h,\Rpm^{(\vep)*}f>\\
=&\ \la\int_{[E_0,\infty)}\int_{[E_0,\infty)}
\dfrac{\mu^{1/2}}{(\mu^2-\mu'^2\mp i\vep)\Dmp(\mu^2-E_0^2)}
d\inner<h,E(\mu)(\ga-1)g>d\inner<T^{1/2}g,E(\mu')f>\\
=&\ \inner<h,(\ga-1)\Rpm^{(\vep)*}f>.
\end{align*}
Thus, by a limit argument, we obtain $A_{-}\Rpm^*=(\ga-1)\Rpm^*$. 
Moreover, \eqref{eq:psigpm} and the equation $(A_{-})^*=-A_{+}$ imply that 
\begin{align*}
\inner<h,A_{-}(A_{-})^*f>&=-(i\pi\la)^2\int_{[E_0,\infty)}
\dfrac{\psi_{g,f}(\mu)\psi_g(\mu)}{\Dp(\mu^2-E_0^2)D_{-}(\mu^2-E_0^2)}d\inner<h,E(\mu)g>\\
&=-i\pi\la\int_{[E_0,\infty)}
\dfrac{(\Dp(\mu^2-E_0^2)-\Dm(\mu^2-E_0^2))\psi_{g,f}(\mu)}
{\Dp(\mu^2-E_0^2)\Dm(\mu^2-E_0^2)}d\inner<h,E(\mu)g>\\
&=-\inner<h,(A_{-})^*f+A_{-}f>.
\end{align*}
Hence the equation $A_{-}(A_{-})^*=-A_{-}-(A_{-})^*$ holds.
\end{proof}

\subsection{Operators $\Ompm$}
In this subsection we consider the bounded operators 
\begin{equation*}
\Ompm:=I+\Rpm.
\end{equation*}
Let $x_0<0$ be the zero of $D(z)$ given in Lemma \ref{proD} (2) and 
\begin{equation*}
\Ub:=\sqrt{\frac{\la}{D'(x_0)}}R_{E_0^2+x_0}(T^2)T^{1/2}g, \ P:=\inner<\Ub,\cdot>\Ub.
\end{equation*}
Then it is easy to see that $\|\Ub\|=1,\Ub\in D(T^{-1})\cap D(T^2)$ and  
\begin{equation*}
T\Ub=\sqrt{\la/D'(x_0)}T^{-1/2}g+(E_0^2+x_0)T^{-1}\Ub.
\end{equation*}
Hence $P$ is a projection operator.
\begin{lemma}
\label{omom=i}
Let $\la\neq\lac$. Then the following equations hold:
\begin{align}
\Ompm^*\Ompm&=I,\label{eq:om*om}\\
\Ompm\Ompm^*&=I-\theta(\la_c-\la)P,\label{eq:omom*}
\end{align}
where $\theta$ is the Heaviside function:
\begin{equation*}
\theta(t)=\left\{
\begin{array}{cl}
1&\text{if} \quad t>0,\\
0&\text{if} \quad t<0. 
\end{array}
\right.
\end{equation*}
\end{lemma}

\begin{rem}
Lemma \ref{omom=i} implies that $\Ompm$ are unitary operators if $\la>\lac$ and partial  isometries with their final subspace $\Ran(I-P)$ if $\la<\lac$. 
\end{rem}

\begin{proof}
\begin{enumerate}
\item
We first prove \eqref{eq:om*om}.\\
It is sufficient to prove that $\Rpm^*\Rpm=-(\Rpm+\Rpm^*)$ hold. 
 For any $f,h\in\msH$ and $\vep>0$,
\begin{align*}
\inner<\Rpm^{(\vep)}h,\Rpm^{(\vep)}f>=&
\ \la^2\int_{[E_0,\infty)}\int_{[E_0,\infty)}d\inner<h,E(\mu')T^{1/2}g>d\inner<T^{1/2}g,E(\mu'')f>\\
&\ \times
\inner<\dfrac{R_{\mu'^2\pm i\vep}(T^2)T^{1/2}g}{\Dpm(\mu'^2-E_0^2)},
\dfrac{R_{\mu''^2\pm i\vep}(T^2)T^{1/2}g}{\Dpm(\mu''^2-E_0^2)}>.
\end{align*}
By the definition of the function $D$, we have 
\begin{equation*}
\la\inner<T^{1/2}g,R_z(T^2)T^{1/2}g>=D(z-E_0^2)-1,\ z\in\C\sla(E_0^2,\infty).
\end{equation*}
By this formula and a resolvent identity, we obtain  
\begin{align*}
\inner<\Rpm^{(\vep)}h,\Rpm^{(\vep)}f>
=&\ \la\int_{[E_0,\infty)}\int_{[E_0,\infty)}d\inner<h,E(\mu')T^{1/2}g>d\inner<T^{1/2}g,E(\mu'')f>\\
&\ \times
\dfrac{D(\mu'^2-E_0^2\mp i\vep)-D(\mu''^2-E_0^2\pm i\vep)}
{(\mu'^2-\mu''^2\mp2i\vep)D_{\mp}(\mu'^2-E_0^2)D_{\pm}(\mu''^2-E_0^2)}.\\
=&\ -\inner<E^{(\vep)}_{\pm}h,\Rpm^{(2\vep)} f>
-\inner<\Rpm^{(2\vep)} h,E^{(\vep)}_{\pm}f>,
\end{align*}
where the operators $E^{(\vep)}_{\pm}$ on $\msH$ are given as follows:
\begin{equation*}
\label{eq:evep}
\inner<h,E^{(\vep)}_{\pm}f>:=\int_{[E_0,\infty)}\frac{D(\mu^2-E_0^2\pm i\vep)}{\Dpm(\mu^2-E_0^2)}d\inner<h,E(\mu)f>,\quad h,f\in\msH.
\end{equation*}
The inequality \eqref{eq:dspmbdd} implies that $E^{(\vep)}_{\pm}$ are bounded for all $0<\vep<\vep_0$. 
Thus, by the Lebesgue dominated convergence theorem, 
we have $\slim_{\vep\dra0}E^{(\vep)}_{\pm}=I$. Hence we obtain that  $\Rpm^*\Rpm=-(\Rpm+\Rpm^*)$. 
\item
We next prove  \eqref{eq:omom*} for $\la\neq\lac$.\\
It is sufficient to prove that $\Rpm\Rpm^*=-(\Rpm+\Rpm^*)-\theta(\la_c-\la)P$ hold.
For any $f,h\in\msH$ and a fixed $\vep>0$, \eqref{eq:rpm*} implies 
\begin{align}
\label{eq:rpmrpm*}
&\inner<\Rpm^{(\vep)*}h,\Rpm^{(\vep)*}f>\nonumber\\
=&\ \la^2\int_{[E_0,\infty)}\int_{[E_0,\infty)}
d\inner<h,E(\mu)T^{1/2}g>d\inner<T^{1/2}g,E(\mu')f>\nonumber\\
&\ \times
\inner<R_{\mu^2\pm i\vep}(T^2)\Dmp(T^2-E_0^2)^{-1}T^{1/2}g,
R_{\mu'^2\pm i\vep}(T^2)\Dmp(T^2-E_0^2)^{-1}T^{1/2}g>\nonumber\\
=&\ \la^2\int_{[E_0,\infty)}\int_{[E_0,\infty)}\int_{[E_0,\infty)}
d\inner<h,E(\mu)T^{1/2}g>d\inner<T^{1/2}g,E(\mu')f>d\|E(\mu'')g\|^2\nonumber\\
&\ \times
\dfrac{\mu''}{(\mu''^2-\mu^2\pm i\vep)(\mu''^2-\mu'^2\mp i\vep)
\Dpm(\mu''^2-E_0^2)D_{\mp}(\mu''^2-E_0^2)}\nonumber\\
=&\ \la\int_{[E_0,\infty)}\int_{[E_0,\infty)}
\dfrac{1}{\mu^2-\mu'^2\mp2i\vep}
J_{\vep}^{\pm}(\mu,\mu')d\inner<h,E(\mu)T^{1/2}g>d\inner<T^{1/2}g,E(\mu')f>,
\end{align}
where, for any $\mu,\mu'\in[E_0,\infty)$,
\begin{align*}
&J_{\vep}^{\pm}(\mu,\mu')\\
=&\int_{[E_0,\infty)}
\dfrac{\la\mu''}{\Dpm(\mu''^2-E_0^2)D_{\mp}(\mu''^2-E_0^2)}
\left(\dfrac{1}{\mu''^2-\mu^2\pm i\vep}-\dfrac{1}{\mu''^2-\mu'^2\mp i\vep}\right)
d\|E(\mu'')g\|^2.
\end{align*}
Then, by change of variable and \eqref{eq:psigpm}, one can show that 
\begin{equation*}
J_{\vep}^{\pm}(\mu,\mu')=\lim_{R\ra\infty}\dfrac{1}{2\pi i}I_{\vep,R}^{\pm}(\mu,\mu'),
\end{equation*}
where, for $R>0$, 
\begin{equation*}
I_{\vep,R}^{\pm}(\mu,\mu')=\int_0^R
\left(\dfrac{1}{\Dp(s)}-\dfrac{1}{\Dm(s)}\right)
G_{\mu,\mu'}^{\vep,\pm}(s)ds
\end{equation*}
and 
\begin{equation*}
G_{\mu,\mu'}^{\vep,\pm}(z):=
\dfrac{1}{z-\mu'^2+E_0^2\mp i\vep}-\dfrac{1}{z-\mu^2+E_0^2\pm i\vep}, \ z\in\C. 
\end{equation*}
For $0<\eta<\vep$ and $R>0$, let $C_{i}\ (i=1,2,3)$ be the curve given as follows:
\begin{equation*}
\begin{array}{ll}
C_1:\ \theta_1(t)=R-t-i\eta,&t: 0 \ra R,\\
C_2:\ \theta_2(t)=\eta e^{-it},& t:\pi/2 \ra (3\pi)/2,\\
C_3:\ \theta_3(t)=t+i\eta,& t: 0 \ra R.
\end{array}
\end{equation*}
Then, for $C=C_1+C_2+C_3$, we have by the Lebesgue dominated convergence theorem, 
\begin{equation*}
I_{\vep,R}^{\pm}(\mu,\mu')=\lim_{\eta\dra0}\int_{C}\dfrac{1}{D(z)}G_{\mu,\mu'}^{\vep,\pm}(z)dz.
\end{equation*}
We take $R$ such that $R>\max\{\mu^2-E_0^2,\mu'^2-E_0^2\}$ and 
define a curve $C_4:\theta_4(t)=\sqrt{\eta^2+R^2} e^{-it},  t: t_s \ra t_f$, 
for $t_s:=\arctan(\eta/R)$ and $t_f=2\pi-t_s$. 
We consider two cases separately.
\begin{enumerate}
\item
The case $\la>\lac$. In this case, the function $G_{\mu,\mu'}^{\vep,\pm}(z)/D(z), z\in\C\sla(0,\infty)$ has two simple poles at 
$z=\mu^2-E_0^2\mp i\vep, z=\mu'^2-E_0^2\pm i\vep$. 
Then, by the residue theorem, we have 
\begin{align*}
\int_{C}\dfrac{1}{D(z)}G_{\mu,\mu'}^{\vep,\pm}(z)dz=&\ 
2\pi i\left(\dfrac{1}{D(\mu'^2-E_0^2\pm i\vep)}-\dfrac{1}{D(\mu^2-E_0^2\mp i\vep)}\right)\\
&-\int_{C_4}\dfrac{1}{D(z)}G_{\mu,\mu'}^{\vep,\pm}(z)dz.
\end{align*}
Thus, as $\eta$ tends to 0, we have  
\begin{align*}
I_{\vep,R}^{\pm}(\mu,\mu')=&\ 
2\pi i\left(\dfrac{1}{D(\mu'^2-E_0^2\pm i\vep)}-\dfrac{1}{D(\mu^2-E_0^2\mp i\vep)}\right)\\
&-\lim_{\eta\dra0}\int_{C_4}\dfrac{1}{D(z)}G_{\mu,\mu'}^{\vep,\pm}(z)dz.
\end{align*}
The definition of line integral implies 
\begin{equation*}
\int_{C_4}\dfrac{1}{D(z)}G_{\mu,\mu'}^{\vep,\pm}(z)dz=
-i\int_{t_s}^{2\pi-t_s}\dfrac{G_{\mu,\mu'}^{\vep,\pm}(\sqrt{\eta^2+R^2} e^{-it})\sqrt{\eta^2+R^2} e^{-it}}{D(\sqrt{\eta^2+R^2} e^{-it})}dt.
\end{equation*}
By the triangle inequality, for any $t\in[t_s,t_f]$, 
\begin{align*}
&|G_{\mu,\mu'}^{\vep,\pm}(\sqrt{\eta^2+R^2} e^{-it})|
\leq
\dfrac{|\mu^2-\mu'^2\pm2i\vep|}{(R-|\mu^2-E_0^2\pm i\vep|)(R-|\mu'^2-E_0^2\mp i\vep|)}.
\end{align*}
On the other hand, by Lemma \ref{Dpm>0}, \eqref{eq:dspmbdd} and the Lebesgue dominated convergence theorem, there are constants $\tilde{R}>0$ and $c_0>0$ such that $|D(z)|\geq c_0$ for all $|z|\geq\tilde{R}$.
Thus we have
\begin{equation*}
I_{\vep,R}^{\pm}(\mu,\mu')=2\pi i\left(\dfrac{1}{D(\mu'^2-E_0^2\pm i\vep)}
-\dfrac{1}{D(\mu^2-E_0^2\mp i\vep)}\right)+O(R^{-1}) \ (R\ra\infty),
\end{equation*}
where $O(\cdot)$ stands for the well known Landau symbol.
Therefore we have 
\begin{equation*}
J_{\vep}^{\pm}(\mu,\mu')
=\dfrac{1}{D(\mu'^2-E_0^2\pm i\vep)}-\dfrac{1}{D(\mu^2-E_0^2\mp i\vep)}
\end{equation*}
for each $\mu,\mu'\in[E_0,\infty)$. Thus, by \eqref{eq:rpmrpm*}, we have
\begin{align*}
\inner<\Rpm^{(\vep)*}h,\Rpm^{(\vep)*}f>=
-\inner<\left(\Rpm^{(2\vep)}\right)^*h,\left(E^{(\vep)}_{\pm}\right)^{-1}f>
-\inner<\left(E^{(\vep)}_{\pm}\right)^{-1}h,\left(\Rpm^{(2\vep)}\right)^*f>.
\end{align*}
As in the proof in (1), we obtain $\slim_{\vep\dra0}\left(E^{(\vep)}_{\pm}\right)^{-1}=I$. Therefore we obtain 
\begin{equation*}
\lim_{\vep\dra0}\inner<\Rpm^{(\vep)*}h,\Rpm^{(\vep)*}f>
=-\inner<\Rpm^*h,f>-\inner<h,\Rpm^*f>. 
\end{equation*}
Thus we obtain the desired result.
\item
The case $\la<\lac$. 
In this case, $G_{\mu,\mu'}^{\vep,\pm}(z)/D(z)$ has a simple pole at $z=x_0$ in addition to
$z=\mu^2-E_0^2\mp i\vep, z=\mu'^2-E_0^2\pm i\vep$. 
The residue $R_0$ of $G_{\mu,\mu'}^{\vep,\pm}(z)/D(z)$ at $z=x_0$ is give by 
\begin{equation*}
R_0=\dfrac{1}{D'(x_0)}\dfrac{\mu'^2-\mu^2\pm2i\vep}
{(x_0-\mu'^2+E_0^2\mp i\vep)(x_0-\mu^2+E_0^2\pm i\vep)}.
\end{equation*}
Thus we have 
\begin{equation*}
J_{\vep}^{\pm}(\mu,\mu')
=\dfrac{1}{D(\mu'^2-E_0^2\pm i\vep)}-\dfrac{1}{D(\mu^2-E_0^2\mp i\vep)}+R_0
\end{equation*}
and also 
\begin{equation*}
\dfrac{\la}{\mu^2-\mu'^2\mp2i\vep}R_0
=-\dfrac{\la}{D'(x_0)}\dfrac{1}
{(\mu'^2-E_0^2-x_0\pm i\vep)(\mu^2-E_0^2-x_0\mp i\vep)}.
\end{equation*}
This implies  that 
\begin{align*}
&\la\lim_{\vep\dra0}\int_{[E_0,\infty)}\int_{[E_0,\infty)}
\dfrac{1}{\mu^2-\mu'^2\mp2i\vep}R_0\ 
d\inner<h,E(\mu)T^{1/2}g>d\inner<T^{1/2}g,E(\mu')f>\\
=&-\inner<h,\Ub>\inner<\Ub,f>=-\inner<h,Pf>,
\end{align*}
Thus we obtain the desired result.
\end{enumerate}
\end{enumerate}
\end{proof}

\subsection{Operators $U$ and $V$}
In this subsection, we investigate the operators $U$ and $V$ defined as follows:
\begin{equation*}
U:=\dfrac{1}{2}(T^{-1/2}\Omp T^{1/2}+T^{1/2}\Omp T^{-1/2}), 
V:=\dfrac{1}{2}(T^{-1/2}\Omp T^{1/2}-T^{1/2}\Omp T^{-1/2}),
\end{equation*}
which are used to construct a Bogoliubov transformation. 
Then, by Lemma \ref{ranrpm}, one can easily see that 
$D(U)=D(V)=D(T^{-1/2})\cap D(T^{1/2})$.
\begin{lemma}
\label{uvbdd}
The operators $U$ and $V$ are bounded.
\end{lemma}

\begin{proof} 
By \eqref{eq:rpm} and Lemma \ref{ranrpm} we have 
\begin{align}
\inner<h,T^{-1/2}\Rpm T^{1/2}f>
&=\dfrac{\la\pi}{2}\int_{[E_0,\infty)}
\left(H\phi_{T^{1/2}g,f}^{\pm}\right)(\mu^2)
\mp i\dfrac{\psi_{g,f}(\mu)}{\Dpm(\mu^2-E_0^2)}\ d\inner<h,E(\mu)g>,\label{eq:t-1/2rpmt1/2}\\
\inner<h,T^{1/2}\Rpm T^{-1/2}f>
&=\dfrac{\la\pi}{2}\int_{[E_0,\infty)}
\left(H\phi_{T^{-1/2}g,f}^{\pm}\right)(\mu^2)\mu \mp i 
\dfrac{\psi_{g,f}(\mu)}{\Dpm(\mu^2-E_0^2)}\ d\inner<h,E(\mu)g>.\label{eq:t1/2rpmt-1/2}
\end{align}
By Assumption \ref{ass} (3), \eqref{eq:hil}, \eqref{eq:muhil} and a property of Hilbert transform,  
we can show that 
\begin{equation*}
\|T^{-1/2}\Rpm T^{1/2}f\|,\|T^{1/2}\Rpm T^{-1/2}f\|\leq\frac{|\la|\pi(C_g+c_g)}{2\delta}\|f\|,
\end{equation*}
where $C_g:=(\sup_{E_0<x}x^{-1}\psi_g(x))^{1/2}(\sup_{E_0<x}x\psi_g(x))^{1/2}$.  
Hence the operators $T^{-1/2}\Rpm T^{1/2}$ and $T^{1/2}\Rpm T^{-1/2}$ are bounded.
\end{proof}

In the same way as in the proof of Lemma \ref{uvbdd}, we see that $T^{-1/2}\Rpm^*T^{1/2}$ and  $T^{1/2}\Rpm^*T^{-1/2}$ are bounded on each domain $D(T^{1/2})$ and $D(T^{-1/2})$. 
In what follows, we write the bounded extension of $U$ and $V$ by the same symbol respectively. Then 
\begin{equation*}
U^*=\frac{1}{2}(\overline{T^{-1/2}\Omp^*T^{1/2}}+\overline{T^{1/2}\Omp^*T^{-1/2}}).
\end{equation*}

\begin{lemma}
\label{uvli}
The operators $U^{\sharp}$ and $V^{\sharp}$ leave $D(T^{-1/2})\ ($resp. $D(T^{1/2}), D(T))$ invariant.
\end{lemma}

\begin{proof}
By applying Lemma \ref{ranrpm} and using the equations 
\begin{equation*}
U^{\sharp}=I+\frac{1}{2}
\left(\overline{T^{-1/2}\Rp^{\sharp}T^{1/2}}+\overline{T^{1/2}\Rp^{\sharp}T^{-1/2}}\right),
\end{equation*}
one can easily see that the assertion for $U^{\sharp}$ is true. 
Similarly one can prove the statement for $V^{\sharp}$.
\end{proof}

\begin{lemma}
\label{omf(T)om}
Let $F(x)=x^{\pm1/2},x^{\pm1}, a.e.\ x\in (0,\infty)$.
Then 
\begin{equation}
\label{eq:omf(T)om}
\Omp F(T)\Omp^*=(\Omp)_J F(T)(\Omp^*)_J\quad \text{on}\ D(F(T)).
\end{equation}
\end{lemma}

\begin{proof}
By Lemma \ref{ranrpm}, the domain of each side of \eqref{eq:omf(T)om} includes $D(F(T))$.  
By Lemmas \ref{AF(T)} and \ref{ARpm}, we have 
\begin{align*}
(\Omp)_J F(T)(\Omp^*)_J
=&\ \Rp F(T)\Rp^*+\Rp\{(A_{-})^*+I\} F(T)\ga+F(T)\ga^*(A_{-}+I)\Rp^*\\
&+F(T)\{A_{-}(A_{-})^*+A_{-}+(A_{-})^*+I\}\\
=&\ \Rp F(T)\Rp^*+\Rp F(T)+F(T)\Rp^*+F(T)\\
=&\ \Omp F(T)\Omp^*.
\end{align*}
\end{proof}

\section{Commutation relations}
In this section, we prove that the pair $(U,V)$ satisfies the condition \eqref{eq:uvcondition}, $V$ is Hilbert-Schmidt and 
\begin{equation*}
B(f):=A(Uf)+A(JVf)^*,\ f\in\msH
\end{equation*}
satisfies some commutation relations with $H(\lambda)$. 
We denote the closure  of $B(f)$ by the same symbol. 
By Lemma \ref{uvli}, we have $D(\SQ(T)^{1/2})\subset D(B(f))\cap D(B(f)^*)$ for all $f\in D(T^{-1/2})$.
\begin{thm}
\label{CR}
The following commutation relations hold:
\begin{enumerate}
\item
For any $f\in D(T)$ and $\psi\in\fbfin(D(T))$,
\begin{equation}
\label{eq:cronfbfin}
[H(\la),B(f)]\psi=-B(Tf)\psi.
\end{equation}
\item
For any $f\in D(T^{-1/2})\cap D(T)$ and $\psi,\phi\in D(\SQ(T))$,
\begin{equation}
\label{eq:cronsqt}
\inner<H(\la)\phi,B(f)\psi>-\inner<B(f)^*\phi,H(\la)\psi>
=-\inner<\phi,B(Tf)\psi>.
\end{equation}
\item
For any $f\in D(T^{-1/2})\cap D(T)$, $B(f)$ maps $D(\SQ(T)^{3/2})$ into $D(\SQ(T))$ and for any $\psi\in D(\SQ(T)^{3/2})$,
\begin{equation}
\label{eq:cronsqt3/2}
[H(\la),B(f)]\psi=-B(Tf)\psi.
\end{equation}
\end{enumerate}
\end{thm}

The both sides of \eqref{eq:cronfbfin},\eqref{eq:cronsqt} and \eqref{eq:cronsqt3/2} have meaning by Lemma \ref{uvli}. 
To prove this theorem, we prove the following lemma:

\begin{lemma}
\label{cr}
For any $f\in D(T)$, the following equations hold:
\begin{align}
[U,T]f=(VT+TV)f&=\dfrac{\la}{2}\inner<\Dm(T^2-E_0^2)^{-1}g, f>g,\label{eq:cruv1}\\
(V^*J-U^*)g&=-\Dm(T^2-E_0^2)^{-1}g.\label{eq:cruv2}
\end{align}
\end{lemma}

\begin{proof}
For any $f,h\in D(T^{-1/2})\cap D(T^{3/2})$, we obtain 
\begin{equation*}
\inner<h,[U,T]f>
=\dfrac{1}{2}\left(\inner<T^{1/2}\Rp^*T^{-1/2}h,Tf>-\inner<Th,T^{1/2}\Rp T^{-1/2}f>\right).
\end{equation*}
Then, for each $\vep>0$, we have  
\begin{align*}
&\inner<T^{1/2}\Rpm^{(\vep)*}T^{-1/2}h,Tf>-\inner<Th,T^{1/2}\Rpm^{(\vep)}T^{-1/2}f>\\
=&\ \la\int_{[E_0,\infty)}\int_{[E_0,\infty)}
\dfrac{\mu'^2-\mu^2}{(\mu'^2-\mu^2\pm i\vep)\Dpm(\mu'^2-E_0^2)}
d\inner<h,E(\mu)g>d\inner<g,E(\mu')f>\\
=&\ \la\int_{[E_0,\infty)}\int_{[E_0,\infty)}
\dfrac{1}{\Dpm(\mu'^2-E_0^2)}
d\inner<h,E(\mu)g>d\inner<E(\mu')g,f>\mp i\vep\inner<T^{-1/2}h,\Rpm^{(\vep)}T^{-1/2}f>.
\end{align*}
Taking the limit $\vep\dra0$, we have 
\begin{equation*}
\inner<T^{1/2}\Rpm^*T^{-1/2}h,Tf>-\inner<Th,T^{1/2}\Rpm T^{-1/2}f>
=\inner<h,\la\inner<\Dmp(T^2-E_0^2)^{-1}g, f>g>.
\end{equation*}
Thus we have 
\begin{equation*}
\inner<h,[U,T]f>=\dfrac{\la}{2}\inner<h,\inner<\Dm(T^2-E_0^2)^{-1}g, f>g>.
\end{equation*}
Since $D(T^{-1/2})\cap D(T^{3/2})$ is a core of $T$, the equation \eqref{eq:cruv1} holds for $f\in D(T)$. To prove \eqref{eq:cruv2}, we note that 
\begin{align*}
(V^*J-U^*)g&=\dfrac{1}{2}(T^{1/2}\Omp^* T^{-1/2}J-T^{-1/2}\Omp^* T^{1/2}J
-T^{1/2}\Omp^*T^{-1/2}-T^{-1/2}\Omp^* T^{1/2})g\\
&=-T^{-1/2}\Omp^* T^{1/2}g,
\end{align*}
where we have used $Jg=g$. Thus, for any $f\in\msH$, we obtain 
\begin{align*}
&\inner<f,(V^*J-U^*)g>\\
=&-\inner<f,g>
-\la\lim_{\vep\dra0}\int_{[E_0,\infty)}
\inner<f,R_{\mu'^2+i\vep}(T^2)D_{-}(T^2-E_0^2)^{-1}g>d\|E(\mu')T^{1/2}g\|^2\\
=&-\inner<f,g>
+\la\lim_{\vep\dra0}\int_{[E_0,\infty)}
\int_{[E_0,\infty)}\dfrac{\mu'}{\mu'^2-\mu^2+ i\vep}d\|E(\mu')g\|^2
\dfrac{1}{\Dm(\mu^2-E_0^2)}\inner<f,E(\mu)g>\\
=&-\inner<f,g>
+\int_{[E_0,\infty)}\dfrac{\Dm(\mu^2-E_0^2)-1}{\Dm(\mu^2-E_0^2)}
d\inner<f,E(\mu)g>\\
=&-\inner<f,\Dm(T^2-E_0^2)^{-1}g>.
\end{align*}
Hence \eqref{eq:cruv2} holds. 
\end{proof}

{\it Proof of Theorem \ref{CR}}.
\begin{enumerate}
\item
By Lemma \ref{uvli}, for any $f\in D(T)$, $B(f)$ leaves $\fbfin(D(T))$ invariant and $H(\la)$ maps $\fbfin(D(T))$ into $\fbfin(\msH)\subset D(B(f))$.
Thus, by using \eqref{eq:ccr}, \eqref{eq:crsq}, we have for any $\psi\in\fbfin(D(T))$,
\begin{equation*}
[H(\la),B(f)]\psi=\left\{-A(TUf)+A(TJVf)^*-\frac{\la}{\sqrt2}\inner<f,(V^*J-U^*)g>\Segal(g)\right\}\psi.
\end{equation*}
Hence by Lemma \ref{cr}, \eqref{eq:cronfbfin} holds. 
\item
By Lemma \ref{uvli} and fundamental properties of the annihilation operators and creation operators, we can see that, for any $f\in D(T^{-1/2})$, $D(\SQ(T)^{1/2})\subset D(B(f))$. 
For any $\psi,\phi\in D(\SQ(T))$, there are sequences 
$\psi_{n},\phi_n\in\fbfin(D(T)), n\in\N$ such that 
$\psi_{n}\ra\psi,\phi_{n}\ra\phi,\SQ(T)\psi_{n}\ra\SQ(T)\psi,\SQ(T)\phi_{n}\ra\SQ(T)\phi$ as $n\ra\infty$, since $\fbfin(D(T))$ is a core of $\SQ(T)$. By (1), we have 
\begin{equation*}
\inner<H(\la)\phi_{n},B(f)\psi_{k}>-\inner<B(f)^*\phi_{n},H(\la)\psi_{k}>
=-\inner<\phi_{n},B(Tf)\psi_{k}>
\end{equation*}
for all $n,k\in\N$ and $f\in D(T^{-1/2})\cap D(T)$. 
By the fundamental inequalities \eqref{eq:ani} and \eqref{eq:cre} and the $\SQ(T)$-boundedness of $\Segal(g)^2$, we can see that $\{B(f)\psi_{n}\}_{n=1}^{\infty}, \{B(f)\phi_{n}\}_{n=1}^{\infty}$, $\{\Segal(g)^2\psi_{n}\}_{n=1}^{\infty}, \{\Segal(g)^2\phi_{n}\}_{n=1}^{\infty}$ and $\{B(Tf)\psi_{n}\}_{n=1}^{\infty}$ converge. 
Hence we obtain \eqref{eq:cronsqt}. 
\item
By Lemma \ref{uvli} and fundamental properties of the annihilation operators and creation operators, we see that, for any $f\in D(T^{-1/2})\cap D(T)$, $B(f)$ maps $D(\SQ(T)^{3/2})$ into $D(\SQ(T))$. Therefore, by \eqref{eq:cronsqt} and the density of $D(\SQ(T))$, we have \eqref{eq:cronsqt3/2}.  $\ \hfill\Box$
\end{enumerate}
\subsection{Relations of $U$ and $V$}

\begin{lemma}
\label{u*u-v*v}
Let $\la\neq\lac$. Then the following equations hold:
\begin{equation}
\label{eq:propercondition}
\left\{
\begin{array}{rl}
U^*U-V^*V&=I,\\
U^*_JV-V^*_JU&=0,\\
UU^*-V_JV^*_J&=I-\theta(\la_c-\la)Q_{+},\\
UV^*-V_JU^*_J&=\theta(\la_c-\la)Q_{-},
\end{array}
\right.
\end{equation}
where 
\begin{equation*}
Q_{\pm}:=\frac{1}{2}\left(\inner<T^{1/2}\Ub,\cdot>T^{-1/2}\Ub\pm
\inner<T^{-1/2}\Ub,\cdot>T^{1/2}\Ub\right)
\end{equation*}
are bounded operators on $\msH$.
\end{lemma}

\begin{proof}
It is sufficient to prove \eqref{eq:propercondition} on $D(T^{-1/2})\cap D(T^{1/2})$. 
Using \eqref{eq:om*om}, one can show that the first equation in \eqref{eq:om*om} hold. 
We have 
\begin{align*}
U_J^* V-V_J^* U
=\dfrac{1}{2}(-T^{1/2}(\Omp^*)_J\Omp T^{-1/2}+T^{-1/2}(\Omp^*)_J\Omp T^{1/2}).
\end{align*}
Multiplying the equation by $(\Omp)_J$ from the left, and using Lemma \ref{omf(T)om}, 
we obtain 
\begin{align*}
(\Omp)_J(U_J^* V-V_J^* U)&=
(\Omp)_J(-T^{1/2}(\Omp^*)_J\Omp T^{-1/2}+T^{-1/2}(\Omp^*)_J\Omp T^{1/2})\\
&=\Omp(-T^{1/2}\Omp^*\Omp T^{-1/2}+T^{-1/2}\Omp^*\Omp T^{1/2})=0.
\end{align*}
By \eqref{eq:om*om}, this implies that $U_J^* V-V_J^* U=0$. 
By Lemma \ref{f(T)_J=f(T)} and Lemma \ref{omf(T)om}, we have 
\begin{align*}
V_JV^*_J=&\ \dfrac{1}{4}\{T^{-1/2}(\Omp T\Omp^*)_J T^{-1/2}
-T^{-1/2}(\Omp\Omp^*)_J T^{1/2}\\
&\quad-T^{1/2}(\Omp\Omp^*)_J T^{-1/2}
+T^{1/2}(\Omp T^{-1}\Omp^*)_J T^{1/2}\}\\
=&\ \dfrac{1}{4}(T^{-1/2}\Omp T\Omp^*T^{-1/2}-T^{-1/2}\Omp\Omp^*T^{1/2}\\
&\quad-T^{1/2}\Omp\Omp^*T^{-1/2}+T^{1/2}\Omp T^{-1}\Omp^*T^{1/2})\\
=&\ VV^*.
\end{align*}
Hence, by direct calculations and \eqref{eq:omom*}, 
one obtains $UU^*-V_JV_J^*=I-\theta(\la_c-\la)Q_{+}$. 
Similarly one can prove the last equation in \eqref{eq:propercondition} (note that $P_J=P$).
\end{proof}

\subsection{Hilbert-Schmidtness of $V$}
In this subsection, we show that $V$ is Hilbert-Schmidt. 
Then we can use Theorem \ref{PC} in the case of $\la>\lac$. 
\begin{lemma}
\label{hs}
The operator $V$ is Hilbert-Schmidt.
\end{lemma}

\begin{proof}
On $D(T^{-1/2})\cap D(T^{1/2})$, $V^*V$ is calculated as follows: 
\begin{align*}
V^*V
=&\ \dfrac{1}{4}(T^{-1/2}\Rp T^{1/2}+T^{1/2}\Rp^* T^{-1/2}+T^{1/2}[\Rp^*,T^{-1}]\Rp T^{1/2}\\
&\quad+T^{1/2}\Rp T^{-1/2}+T^{-1/2}\Rp^* T^{1/2}+T^{-1/2}[\Rp^*,T]\Rp T^{-1/2}\\
&\quad+T^{1/2}\Rp^*\Rp T^{-1/2}+T^{-1/2}\Rp^*\Rp T^{1/2})\\
=&\ \dfrac{1}{4}(T^{1/2}[\Rp^*,T^{-1}]\Rp T^{1/2}+T^{-1/2}[\Rp^*,T]\Rp T^{-1/2}),
\end{align*}
where we have used the formula $\Rp^*\Rp=-(\Rp+\Rp^*)$ in the proof of Lemma \ref{omom=i} and Lemma \ref{ranrpm}. 
Thus, for any $f\in D(T^{-1/2})\cap D(T^{1/2})$ and $\vep>0$, we have 
\begin{align*}
&\inner<f,(T^{1/2}[\Rp^{(\vep)*},T^{-1}]\Rp^{(\vep)} T^{1/2}+T^{-1/2}
[\Rp^{(\vep)*},T]\Rp^{(\vep)} T^{-1/2})f>\\
=&\ \la\int_{[E_0,\infty)}\int_{[E_0,\infty)}
\dfrac{\mu'}{(\mu'^2-\mu^2+i\vep)\Dp(\mu'^2-E_0^2)}
d\inner<[T^{-1},\Rp^{(\vep)}]T^{1/2}f,E(\mu)T^{1/2}g>d\inner<E(\mu')g,f>\\
&+\la\int_{[E_0,\infty)}\int_{[E_0,\infty)}
\dfrac{\mu}{(\mu'^2-\mu^2+i\vep)\Dp(\mu'^2-E_0^2)}
d\inner<[T,\Rp^{(\vep)}]T^{-1/2}f,E(\mu)T^{-1/2}g>d\inner<E(\mu')g,f>.
\end{align*}
Then, for any $B\in\Borel$, we can see 
\begin{align}
&\inner<[T^{-1},\Rp^{(\vep)}]T^{1/2}f,E(B)T^{1/2}g>\nonumber\\
=&\ \la\int_{B}\int_{[E_0,\infty)}
\dfrac{\mu''-\mu}{(\mu''^2-\mu^2-i\vep)\Dm(\mu''^2-E_0^2)}
d\inner<f,E(\mu'')g>d\|E(\mu)g\|^2.\label{eq:cvformulaofLSI}
\end{align}
Similarly, we obtain 
\begin{align*}
&\inner<[T,\Rp^{(\vep)}]T^{-1/2}f,E(B)T^{-1/2}g>\\
=&\ \la\int_{B}\int_{[E_0,\infty)}
\dfrac{\mu-\mu''}{(\mu''^2-\mu^2-i\vep)\Dm(\mu''^2-E_0^2)}
d\inner<f,E(\mu'')g>d\|E(\mu)g\|^2.
\end{align*}
Thus, by a formula of change of variable in Lebesgue-Stieltjes integration and Fubini's theorem, 
we have 
\begin{align*}
&\inner<f,(T^{1/2}[\Rp^{(\vep)*},T^{-1}]\Rp^{(\vep)} T^{1/2}+T^{-1/2}
[\Rp^{(\vep)*},T]\Rp^{(\vep)} T^{-1/2})f>\\
=&\ \la^2\int_{[E_0,\infty)}\int_{[E_0,\infty)}\int_{[E_0,\infty)}
d\|E(\mu)g\|^2d\inner<f,E(\mu'')g>d\inner<E(\mu')g,f>\\
&\quad\times\dfrac{(\mu-\mu')(\mu-\mu'')
}{(\mu'^2-\mu^2+i\vep)(\mu''^2-\mu^2-i\vep)\Dp(\mu'^2-E_0^2)D_{-}(\mu''^2-E_0^2)}.
\end{align*}
Then it is easy to see that for any $\mu,\mu',\mu''\in[E_0,\infty)$ 
\begin{align*}
&\lim_{\vep\dra0}\dfrac{(\mu-\mu')(\mu-\mu'')
}{(\mu'^2-\mu^2+i\vep)(\mu''^2-\mu^2-i\vep)\Dp(\mu'^2-E_0^2)D_{-}(\mu''^2-E_0^2)}\\
=&\ \dfrac{1}{(\mu'+\mu)(\mu''+\mu)\Dp(\mu'^2-E_0^2)D_{-}(\mu''^2-E_0^2)}.
\end{align*}
For any $\vep>0$ and $\mu,\mu',\mu''\in[E_0,\infty)$, we have by Lemma \ref{Dpm>0} and  the arithmetic-geometric mean inequality
\begin{equation*}
\left|\dfrac{(\mu-\mu')(\mu-\mu'')
}{(\mu'^2-\mu^2+i\vep)(\mu''^2-\mu^2-i\vep)\Dp(\mu'^2-E_0^2)D_{-}(\mu''^2-E_0^2)}\right|
\leq\dfrac{1}{4\delta^2\mu\sqrt{\mu'\mu''}}.
\end{equation*}
On the other side, for any $\alpha, \beta\in\C$, we see 
\begin{align*}
&\int_{[E_0,\infty)}\int_{[E_0,\infty)}\int_{[E_0,\infty)}
\frac{1}{\mu\sqrt{\mu'\mu''}}
d\|E(\mu)g\|^2d\|E(\mu'')(f+\alpha g)\|^2 d\|E(\mu')(f+\beta g)\|^2\\
=&\ \|T^{-1/2}g\|^2\|T^{-1/4}(f+\alpha g)\|^2\|T^{-1/4}(f+\beta g)\|^2<\infty.
\end{align*}
Thus, by the Lebesgue dominated convergence theorem, we have 
\begin{align*}
&\lim_{\vep\dra0}\int_{[E_0,\infty)}\int_{[E_0,\infty)}\int_{[E_0,\infty)}
d\|E(\mu)g\|^2d\|E(\mu'')(f+\alpha g)\|^2 d\|E(\mu')(f+\beta g)\|^2\\
&\qquad\times\dfrac{(\mu-\mu')(\mu-\mu'')
}{(\mu'^2-\mu^2+i\vep)(\mu''^2-\mu^2-i\vep)\Dp(\mu'^2-E_0^2)D_{-}(\mu''^2-E_0^2)}\\
=&\int_{[E_0,\infty)}\int_{[E_0,\infty)}\int_{[E_0,\infty)}
d\|E(\mu)g\|^2d\|E(\mu'')(f+\alpha g)\|^2 d\|E(\mu')(f+\beta g)\|^2\\
&\qquad\times\dfrac{1}{(\mu'+\mu)(\mu''+\mu)\Dp(\mu'^2-E_0^2)D_{-}(\mu''^2-E_0^2)}.
\end{align*}
In particular, for each $\alpha,\beta=\pm1,\pm i$, the polarization identity and Fubini's theorem give 
\begin{align*}
\inner<f,V^*Vf>
=\dfrac{\la^2}{4}\int_{[E_0,\infty)}
\left|\inner<f,R_{-\mu}(T)\Dm(T^2-E_0^2)^{-1}g>\right|^2d\|E(\mu)g\|^2.
\end{align*}
Let $\{e_n\}_{n=1}^{\infty}\subset D(T^{-1/2})\cap D(T^{1/2})$ be a CONS of $\msH.$
The termwise integration implies that 
\begin{align}
\sum_{n=1}^{\infty}\inner<e_n,V^*V e_n>
=&\dfrac{\la^2}{4}\int_{[E_0,\infty)}
\|R_{-\mu}(T)\Dm(T^2-E_0^2)^{-1}g\|^2 d\|E(\mu)g\|^2\nonumber\\
\label{eq:trvv}
=&\dfrac{\la^2}{4}\int_{[E_0,\infty)}\int_{[E_0,\infty)}
\dfrac{1}{(\mu'+\mu)^2|\Dm(\mu'^2-E_0^2)|^2}d\|E(\mu')g\|^2d\|E(\mu)g\|^2\\
\leq&\dfrac{\la^2}{16\delta^2}\int_{[E_0,\infty)}\int_{[E_0,\infty)}
\dfrac{1}{\mu'\mu}d\|E(\mu')g\|^2d\|E(\mu)g\|^2<\infty,\nonumber
\end{align}
where we have used the arithmetic-geometric mean and Lemma \ref{Dpm>0}. 
Hence $V$ is Hilbert-Schmidt.
\end{proof}

\begin{lemma}
\label{UB(f)U}
If $\la>\lac$, then there is a unitary operator $\U$ on $\fb(\msH)$ such that for all $f\in \msH$,
\begin{equation*}
\U B(f)\U^{-1}=A(f).
\end{equation*}
\end{lemma}
\begin{proof}
By Lemma \ref{u*u-v*v} and Lemma \ref{hs}, we can apply Theorem \ref{PC}.
\end{proof}

\section{Analysis in the case $\la>\lac$}
In this section we prove Theorem \ref{Main} (1). Before starting the proof, we need to know a  property of the Hamiltonian $H(\lambda)$. 
\subsection{Time evolution}
\begin{thm}[Time evolution]
\label{TE}
If $\la>\lacz$, then for all $f\in D(T^{-1/2})$, $\psi\in D(\SQ(T)^{1/2})$ and $t\in\R$,
\begin{align}
e^{it H(\la)}B(f)e^{-it H(\la)}\psi&=B(e^{itT}f)\psi,\label{eq:teani}\\
e^{it H(\la)}B(f)^*e^{-it H(\la)}\psi&=B(e^{itT}f)^*\psi.\label{eq:tecre}
\end{align}
\end{thm}

\begin{proof}
It is sufficient to prove \eqref{eq:teani}, because \eqref{eq:tecre} follows from taking the adjoint of \eqref{eq:teani}. We define a function $v:\R\ra\C$ by 
$v(t):=\inner<\phi,e^{it H(\la)}B(e^{-itT}f)e^{-itH(\la)}\psi>, t\in\R$ 
for any $f\in D(T^{-1/2})\cap D(T)$ and $\psi, \phi\in D(\SQ(T))$.
Then $v$ is well-defined by operational calculus and Theorem \ref{sa}. 
The function $v$ is differentiable and, by Theorem \ref{CR} (2), we have for any $t\in\R$,
\begin{align*}
\dfrac{d}{dt}v(t)=&
i\inner<H(\la)e^{-it H(\la)}\phi,B(e^{-itT}f)e^{-it H(\la)}\psi>
-i\inner<B(e^{-itT}f)^*e^{-it H(\la)}\phi,H(\la)e^{-it H(\la)}\psi>\\
&+i\inner<e^{-it H(\la)}\phi,B(Te^{-itT}f)e^{-it H(\la)}\psi>\\
=&0.
\end{align*}
Hence $v(t)=v(0)$ for all $t\in\R$. Hence the equation 
\begin{equation*}
\inner<\phi,e^{it H(\la)}B(e^{-itT}f)e^{-it H(\la)}\psi>=\inner<\phi,B(f)\psi>
\end{equation*}
holds for all $t\in\R$. 
By replacing $f$ with $e^{itT}f$, one has for all $\psi\in D(\SQ(T))$,
\begin{equation*}
e^{it H(\la)}B(f)e^{-it H(\la)}\psi=B(e^{itT}f)\psi.
\end{equation*}
Since $D(\SQ(T))$ is a core of $(H(\la)+M)^{1/2}$ and $D(H(\la)+M)^{1/2}=D(\SQ(T)^{1/2})$ by Theorem \ref{sa} (3), 
we obtain \eqref{eq:teani} for $f\in D(T^{-1/2})\cap D(T)$ and $\psi\in D(\SQ(T)^{1/2})$. 
Finally we extend \eqref{eq:teani} for all $f\in D(T^{-1/2})$. Let  $f\in D(T^{-1/2})$ and $\psi\in D(\SQ(T)^{1/2})$. Then we set 
$f_n:=E((-\infty,n])f$ for each $n\in\N$. Then $f_n\in D(T^{-1/2})\cap D(T)$ for all $n\in\N$ and one can easily show that $f_n\ra f,\ T^{-1/2}f_n\ra T^{-1/2}f$ as $n\ra\infty$ by using functional calculus and the the Lebesgue dominated convergence theorem. 
Thus we have $Uf_n\ra Uf,\  JVf_n\ra JVf$ as $n\ra\infty$ by the boundedness of $U$ and $V$. 
By using the linearity of the Hilbert transform and that of the map $f\mapsto\psi_{g,f}$,  \eqref{eq:t-1/2rpmt1/2} and \eqref{eq:t1/2rpmt-1/2}, and \eqref{eq:hil}, 
we can show that $T^{-1/2}Uf_n\ra T^{-1/2}Uf,\ T^{-1/2}JVf_n\ra T^{-1/2}JVf$ as $n\ra\infty$. 
Therefore we obtain $B(f_n)\phi\ra B(f)\phi$ and $B(e^{itT}f_n)\phi\ra B(e^{itT}f)\phi$ as $n\ra\infty$ for any $\phi\in D(\SQ(T)^{1/2})$ by \cite[Lemma4-28]{AR1}. By the preceding result, we have for any $n\in\N$, 
\begin{equation*}
B(f_n)e^{-it H(\la)}\psi=e^{-it H(\la)}B(e^{itT}f_n)\psi.
\end{equation*}
The equation $D(\SQ(T)^{1/2})=D((H(\la)+M)^{1/2})$ in Theorem \ref{sa} (3) implies that 
\begin{equation*}
e^{-it H(\la)}D(\SQ(T)^{1/2})=D(\SQ(T)^{1/2}).
\end{equation*} 
Hence, taking the limit $n\ra\infty$, we obtain \eqref{eq:teani} for $f\in D(T^{-1/2}),\psi\in D(\SQ(T)^{1/2})$.
\end{proof}

\subsection{Proof of Theorem \ref{Main} (1)}
In this subsection, we assume that $\lambda>\lambda_{\rm c}$.
\begin{lemma}
\label{gse}
Let $\Om:=\U^{-1}\Omz$, where $\U$ is the unitary operator in Lemma \ref{UB(f)U} and $\Omz:=(1,0,0,\ldots)\in\fb(\msH)$ is the Fock vacuum. Then 
there is an eigenvalue $\Eg$ of $H(\la)$ and $\Om$ is a corresponding eigenvector: 
$H(\lambda)\Omega=\Eg \Omega$.
\end{lemma}

\begin{proof}
In general, by \cite[Proposition4-4]{AR1} for a dense subspace $\msD\subset\msH$,  if  $\psi\in\cap_{f\in\msD}D(A(f))$ satisfies $A(f)\psi=0$ for all $f\in\msD$, 
then there is a constant $\alpha\in\C$ such that $\psi=\alpha\Omega_0$.
Thus, by Lemma \ref{UB(f)U}, $B(f)\phi=0$ for all $f\in D(T^{-1/2})$. 
Hence there is a constant $\alpha\in\C$ such that $\phi=\alpha\Omega$. 
For any $f\in D(T^{-1/2})$ and $t\in\R$,
\begin{equation*}
B(f)e^{-it H(\la)}\Omega=e^{-it H(\la)}B(e^{itT}f)\Omega=0
\end{equation*}
by Lemma \ref{TE}. 
Thus, for each $t\in\R$, there is a constant $\alpha(t)\in\C$ such that 
$e^{-it H(\la)}\Omega=\alpha(t)\Omega$. 
Then we have $|\alpha(t)|=1, \alpha(t+s)=\alpha(t)\alpha(s)$ for all $t,s\in\R$, 
since $\{e^{-it H(\la)}\}_{t\in\R}$ is a strongly continuous one-parameter unitary group. 
Thus there exists a constant $\Eg\in\R$ such that $\alpha(t)=e^{-it\Eg}, t\in\R$.
The differentiation of the equation $e^{-it H(\la)}\Omega=e^{-it\Eg}\Omega$ in $t$ implies  $\Om\in D(H(\la))$ and $\Om\in\Ker(H(\la)-\Eg)$. 
\end{proof}
{\it Proof of Theorem \ref{Main} (1).}\\
The subspace $\mU:=\msL\{B(f_1)^*\cdots B(f_n)^*\Omega, \Omega\ |\ f_j\in D(T^{-1/2}), j=1,\ldots,n, \ n\in\N\}$ is dense in $\fb(\msH)$ by the fact that 
 $\mU=\U^{-1}\fbfin(D(T^{-1/2}))$, where $\msL(\msD)$ denotes the subspace algebraically spanned by the vectors in $\msD$. 
By Lemma \ref{TE} and Lemma \ref{te}, for any $t\in\R, f_j\in D(T^{-1/2}), j=1,\ldots,n$, we have 
\begin{align*}
e^{it H(\la)}B(f_1)^*\cdots B(f_n)^*\Omega
=&B(e^{itT}f_1)^*\cdots B(e^{itT}f_n)^*e^{it H(\la)}\Omega\\
=&B(e^{itT}f_1)^*\cdots B(e^{itT}f_n)^*e^{it\Eg}\Omega\\
=&e^{it\Eg}\U^{-1}e^{it\SQ(T)}A(f_1)^*\cdots A(f_n)^*\Omega_0\\
=&\U^{-1}e^{it(\SQ(T)+\Eg)}\U B(f_1)^*\cdots B(f_n)^*\Omega.
\end{align*}
By this equation and a limiting argument, we obtain $\U e^{it H(\la)}\U^{-1}=e^{it(\SQ(T)+\Eg)}$. 
By the unitary covariance of functional calculus, we have 
\begin{equation*}
\U e^{it H(\la)}\U^{-1}=e^{it\U H(\la)\U^{-1}},\quad t\in\R.
\end{equation*}
Hence \eqref{eq:main1} holds. 
The equation \eqref{eq:main1} and the well-known spectral properties of $\SQ(T)$ imply that 
$\Eg$ is the ground state energy of $H(\lambda)$ and $\Om$ is the unique ground state 
of $H(\lambda)$. ${}\hfill\Box$

\begin{lemma}
The ground state energy $\Eg$ is given as follows: 
\begin{align}
&\Eg=\dfrac{\la}{4}\|g\|^2-\Tr(T^{1/2}V^*VT^{1/2}),\label{eq:eg}\\
&\Tr(T^{1/2}V^*VT^{1/2})=
\dfrac{\la^2}{4}\disp\int_{[E_0,\infty)}\disp\int_{[E_0,\infty)}
\dfrac{\mu}{(\mu+\mu')^2|\Dm(\mu^2-E_0^2)|^2}d\|E(\mu)g\|^2d\|E(\mu')g\|^2.\label{eq:gse}
\end{align}
\end{lemma}

\begin{proof}
The operator $\U$ leaves $D(\SQ(T))$ invariant by Theorem \ref{Main} (1). 
In particular, $\U\Omz\in D(\SQ(T)^{1/2})$. 
Thus, by Lemma \ref{conv}, the isometry of $\U$ and the definition of $B(\cdot)$, we have $\inner<\Omz,(H(\la)-\Eg)\Omz>=\Tr(T^{1/2}V^*VT^{1/2})$. 
It is easy to see that $\inner<\Omz,H(\la)\Omz>=\la\|g\|^2/4$. 
Hence \eqref{eq:eg} holds. 
Formula \eqref{eq:gse} can be proved in the same way as in \eqref{eq:trvv}.
\end{proof}

\section{Analysis in the case $\lacz<\la<\lac$}
In Section 5, we proved Theorem \ref{Main} (1). But the proof is valid only for the case  $\la>\lac$. Therefore it is necessary to find another pair of  operators $U$ and $V$ if one  wants to use a Bogoliubov transformation for the spectral analysis of $H(\lambda)$ in the case  $\lambda\leq\lambda_{\rm c}$. 
In this section we assume that $T$ and $g$ satisfy Assumption \ref{ass}, $E_0>0$ and 
$\lacz<\la<\lac$. Under these conditions, we can define operators $\xi, X, Y$ and $T_{\pm}$ as follows:
\begin{align*}
&\xi:=\Omp T\Omp^*+\beta P,\\
&X:=U\Omp^*+\Tp P,\ Y:=V\Omp^*+\Tm P,\\
&\Tpm:=\frac{1}{2}(\beta^{1/2}T^{-1/2}\pm\beta^{-1/2}T^{1/2}),
\end{align*}
where $\beta:=(E_0^2+x_0)^{1/2}$.

\begin{rem}
\label{betasmall}
The definition of $x_0$ implies the following:
\begin{equation*}
E_0^2+x_0\left\{
\begin{array}{ccc}
>0,&\mathrm{if}\ &\lacz<\la<\lac,\\
=0,&\mathrm{if}\ &\la=\lacz,\\
<0,&\mathrm{if}\ &\la<\lacz.
\end{array}
\right.
\end{equation*}
Thus, in the case $\lacz<\la<\lac$, we see that the inequality $0<\beta<E_0$ holds. 
Let 
\begin{equation*}
C(f):=A(Xf)+A(JYf)^*, f\in\msH.
\end{equation*}
Then $C(f)$ is a densely defined closable operator. We denotes its closure by the same symbol.
\end{rem}

\subsection{Properties of $X,Y$ and $\xi$}
In this subsection we study operators $X,Y$ and $\xi$. First, we consider $\xi$. Let be 
\begin{equation*}
\tilde{T}:=\Omp T\Omp^*.
\end{equation*}
\begin{lemma}
The operator $\tilde{T}$ is a self-adjoint operator with $D(\tilde{T})=D(T)$.
\end{lemma}

\begin{proof}
By Lemma \ref{ranrpm} we see that $D(\tilde{T})=D(T)$. Hence $\tilde{T}$ is symmetric. 
For any $\phi\in D((\tilde{T})^*)$ and $\psi\in D(T)=D(\tilde{T})$, we have 
$\inner< \Omp^*(\tilde{T})^*\phi, \psi>=\inner<\Omp^*\phi,T\psi>$. 
This implies that $\Omp^*\phi\in D(T)$. Hence $\tilde{T}$ is self-adjoint.
\end{proof}

\begin{lemma}
\label{sppT'}
The spectra of $\tilde{T}$ are as follows:
\begin{equation*}
\sigma(\tilde{T})=\{0\}\cup\sigma(T),\sigac(\tilde{T})=\sigma(T),\sigp(\tilde{T})=\{0\},\sigsc(\tilde{T})=\emptyset.
\end{equation*}
\end{lemma}

\begin{proof}
We define a family of projection operators $\{E_P(B)\ |\ B\in\Borel\}$ on $\msH$ as follows: 
$E_P(B)=0$ if $0\nin B$ and $E_P(B)=P$ if $0\in B$ for each $B\in\Borel$. 
It is easy to see that $\{E_{\tilde{T}}(B):=\Omp E(B)\Omp^*+E_P(B) |\ B\in\Borel\}$ is a  spectral measure. Using functional calculus, we see that $E_{\tilde{T}}(\cdot)$ 
 is the spectral measure of $\tilde{T}$. 
It is easy to see that the absolutely continuous part (resp. singular part) of $\tilde{T}$ is $\tilde{T}\rest\Ran(I-P)$ (resp. $\tilde{T}\rest\Ran(P)$) since $T$ is absolutely continuous and $\Ompm$ are partial isometries. 
Thus we see $\sigma(\tilde{T})=\{0\}\cup\sigac(\tilde{T}),\sigp(\tilde{T})=\{0\},\sigsc(\tilde{T})=\emptyset$.

We next show that $\sigac(\tilde{T})=\sigma(T)$. 
For any $\mu\in\sigma(T)$, there is a sequence $\psi_n\in D(T), n\in\N$ such that $\|\psi_n\|=1$ for all $n\in\N$ and $\lim_{n\ra\infty}\|(T-\mu)\psi_n\|=0$. For each $n\in\N$, there is a  $\phi_n\in$ Ran$(I-P)$  such that $\psi_n=\Omp^*\phi_n$. Then $\|\phi_n\|=\|\Omp\psi_n\|=\|\psi_n\|=1$ and $\|(\tilde{T}-\mu)\phi_n\|=\|(T-\mu)\psi_n\|\ra0$ as $n\ra\infty$. Thus we have $\mu\in\sigma(\tilde{T}\rest\Ran(I-P))=\sigma_{\rm ac}(\tilde{T})$. 
For any $\mu\in\sigma_{\rm ac}(\tilde{T})$, there is a sequence $\eta_n\in D(\tilde{T})\cap\Ran(I-P)$ such that $\|\eta_n\|=1$ and $\lim_{n\ra\infty}\|(\tilde{T}-\mu)\eta_n\|=0$. 
Then we easily see that $\Omp^*\eta_n\in D(T)$ for all $n\in\N$. 
The equation $\Omp\Omp^*\eta_n=\eta_n$ implies that $\|\Omp^*\eta_n\|=1$ for all $n\in\N$ and 
\begin{equation*}
\|(T-\mu)\Omp^*\eta_n\|=\|(\tilde{T}-\mu)\eta_n\|\ra0,\quad n\ra\infty.
\end{equation*}
Thus $\mu\in\sigma(T)$. Hence $\sigac(\tilde{T})=\sigma(T)$. 
\end{proof}

\begin{lemma}
\label{xi}
The operator $\xi$ is an injective, non-negative self-adjoint operator with $D(\xi)=D(T)$ and we have the following equations:
\begin{equation}
\label{eq:sppxi}
\sigma(\xi)=\{\beta\}\cup\sigma(T), \sigac(\xi)=\sigma(T), \sigp(\xi)=\{\beta\}, \sigsc(\xi)=\emps.
\end{equation}
In particular, $\beta$ is the ground state energy of $\xi$, which is an isolated eigenvalue of $\xi$, and $\Ub$ is the unique ground state of $\xi$. 
\end{lemma}

\begin{proof}
By Lemma \ref{sppT'} and a spectral property of direct sum of self-adjoint operators, we have the equation \eqref{eq:sppxi}. 
Thus $\beta$ is an isolated ground state energy by Remark  \ref{betasmall}. 
It is easy to see that $\Ub$ is a ground state of $\xi$. Assume that $f\in \Ker(\xi-\beta)$ satisfies $(I-P)f\neq0$. 
Then $\Omp^*f\neq0$ by Lemma \ref{omom=i}. This implies $T\Omp^*f=\beta\Omp^*f$, but this contradicts Assumption \ref{ass} (1). Hence $(I-P)f=0$, implying that the ground state of $\xi$ is unique.
\end{proof}

\begin{lemma}
\label{xi1/2}
The operators $\xi^{\pm1/2}$ are given by 
\begin{align}
\label{eq:xi1/2}
\xi^{1/2}=&\Omp T^{1/2}\Omp^*+\beta^{1/2} P,\\
\label{eq:xim1/2}
\xi^{-1/2}=&\Omp T^{-1/2}\Omp^*+\beta^{-1/2} P
\end{align}
with $D(\xi^{\pm1/2})=D(T^{\pm1/2})$. 
\end{lemma}

\begin{proof}
We can show in the same way as in the proof of Lemma \ref{xi} that the right hand side of \eqref{eq:xi1/2} is non-negative, self-adjoint operator with its domain $D(T^{1/2})$. 
We have $\xi\subset(\Omp T^{1/2}\Omp^*+\beta^{1/2} P)^2$. Since a self-adjoint operator has no non-trivial symmetric extension, \eqref{eq:xi1/2} holds. 
In the same way as in the case of \eqref{eq:xi1/2}, we can show that the right hand side of  \eqref{eq:xim1/2} is a self-adjoint operator. 
We have $D(\Omp T^{-1/2}\Omp^*+\beta^{-1/2} P)\subset\Ran(\xi^{1/2})$ and $\xi^{1/2}(\Omp T^{-1/2}\Omp^*+\beta^{-1/2} P)=I$ on  $D(\Omp T^{-1/2}\Omp^*)$. Hence 
$\Omp T^{-1/2}\Omp^*+\beta^{-1/2} P\subset\xi^{-1/2}$. 
Thus the equation \eqref{eq:xim1/2} holds.
\end{proof}

Next, we study $X$ and $Y$.
\begin{lemma}
\label{ranxy}
The operators $X^{\sharp}$ and $Y^{\sharp}$ leave $D(T^{-1/2})\ ($resp. $D(T^{1/2}), D(T))$ invariant.
\end{lemma}

\begin{proof}
The assertion follows from Lemma \ref{ranrpm}, Lemma \ref{uvli}, Lemma \ref{xi1/2} and the definition of $X$ and $Y$. 
\end{proof}

\begin{lemma}
The following equations hold:
\begin{equation}
\label{eq:properconditionxy}
\left\{
\begin{array}{rl}
X^*X-Y^*Y&=I,\\
X^*_JY-Y^*_JX&=0,\\
XX^*-Y_JY^*_J&=I,\\
XY^*-Y_JX^*_J&=0.
\end{array}
\right.
\end{equation}
\end{lemma}

\begin{proof}
The operator $P$ (resp. $\Tpm$) satisfies $P_J=P$ ( resp. $(\Tpm)_J=\Tpm$). 
By \eqref{eq:omom*}, we have $\Omp^*\Ub=0$. 
Hence we obtain $(U^*\pm V^*)T^{\pm1/2}\Ub=0$ and $(U^*\Tpm-V^*\Tmp)\Ub=0$. 
The equations $\Tp\Tp-\Tm\Tm=I$ and $\Tp\Tm-\Tm\Tp=0$ hold on $D(T^{-1})\cap D(T)$. 
By \eqref{eq:propercondition} and direct calculations, we have 
$X^*X-Y^*Y=I$ and $X_J^*Y-Y_J^*X=0$. 
By similar calculations, we have $XX^*-Y_JY_J^*=I$ and $XY^*-Y_JX_J^*=0$ on $D(T^{-1/2})\cap D(T^{1/2})$. Then, by a limiting argument, we obtain  \eqref{eq:properconditionxy}.
\end{proof}

\begin{lemma}
\label{hsy}
The operator $Y$ is Hilbert-Schmidt.
\end{lemma}

\begin{proof}
We can easily show that the assertion follows from Lemma \ref{hs}, Lemma \ref{ranxy} and the choice a CONS $\{e_n\}^{\infty}_{n=0}\subset D(T^{-1/2})\cap D(T^{1/2})$ with $e_0=\Ub$. 
\end{proof}

\begin{lemma}
\label{UC(f)U}
There is a unitary operator $\V$ on $\fb(\msH)$ such that for all $f\in \msH$,
\begin{equation*}
\V C(f)\V^{-1}=A(f).
\end{equation*}
\end{lemma}

\begin{proof}
By Theorem \ref{PC}, \eqref{eq:properconditionxy} and Lemma \ref{hsy}, we can prove this assertion.
\end{proof}
\subsection{Commutation relations}
\begin{thm}
\label{CR:C}
The following commutation relations hold:
\begin{enumerate}
\item
For any $f\in D(T)$ and $\psi\in\fbfin(D(T))$,
\begin{equation*}
[H(\la),C(f)]\psi=-C(\xi f)\psi.
\end{equation*}
\item
For any $f\in D(T^{-1/2})\cap D(T)$ and $\psi,\phi\in D(\SQ(T))$,
\begin{equation*}
\label{eq:cronsqt:C}
\inner<H(\la)\phi,C(f)\psi>-\inner<C(f)^*\phi,H(\la)\psi>
=-\inner<\phi,C(\xi f)\psi>.
\end{equation*}
\item
For any $f\in D(T^{-1/2})\cap D(T)$, $C(f)$ maps $D(\SQ(T)^{3/2})$ into $D(\SQ(T))$ and for any $\psi\in D(\SQ(T)^{3/2})$,
\begin{equation*}
\label{eq:cronsqt3/2:C}
[H(\la),C(f)]\psi=-C(\xi f)\psi.
\end{equation*}
\end{enumerate}
\end{thm}

Theorem \ref{CR:C} follows, in the same manner as in the proof of Theorem \ref{CR}, 
from Lemma \ref{uvli}, Lemma \ref{xi1/2} and the next lemma:

\begin{lemma}
\label{crxy}
For any $f\in D(T)$ the following equations hold:
\begin{align}
\label{eq:crxy1}
-TXf+\dfrac{\la}{2}\inner<(Y^*J-X^*)g, f>g=&-X\xi f,\\
\label{eq:crxy2}
TJYf+\dfrac{\la}{2}\inner<f,(Y^*J-X^*)g>g=&-JY\xi f.
\end{align}
\end{lemma}

\begin{rem}
By Lemma \ref{uvli} and the definition of $\xi$, 
the both sides of \eqref{eq:crxy1} and \eqref{eq:crxy2} have meaning.
\end{rem}

\begin{proof}
Let be $a:=\sqrt{\la/D'(x_0)}$. Then we can see by the definition of $x_0$ and \eqref{eq:cruv2},
\begin{align*}
(Y^*J-X^*)g=-\Omp\Dm(T^2-E_0^2)^{-1}g+\frac{\beta^{-1/2}a}{\la}\Ub.
\end{align*}
We have 
\begin{align}
T\Tpm\Ub=&\frac{1}{2}(\beta^{1/2}T^{1/2}\Ub\pm\beta^{-1/2}T^{3/2}\Ub)\nonumber\\
=&\frac{1}{2}(\beta^{1/2}T^{1/2}\Ub\pm\beta^{3/2}T^{-1/2}\Ub\pm\beta^{-1/2}ag). \label{eq:ttpmub}
\end{align}
Thus, for any $f\in D(T)$, we have  
\begin{align*}
&-TXf+\dfrac{\la}{2}\inner<(Y^*J-X^*)g, f>g\\
=&-TU\Omp^*f-\dfrac{\la}{2}\inner<\Dm(T^{2}-E_0^2)^{-1}g,\Omp^*f>g
-T\Tp Pf +\frac{\beta^{-1/2}a}{2}\inner<\Ub,f>g.
\end{align*}
Then, by \eqref{eq:cruv1} and \eqref{eq:ttpmub}, we have 
\begin{align*}
-TXf+\dfrac{\la}{2}\inner<(Y^*J-X^*)g, f>g
=&-UT\Omp^*f-\beta\inner<\Ub,f>\Tp\Ub\\
=&-X(\Omp T\Omp^*+\beta P)f.
\end{align*}
Thus we obtain \eqref{eq:crxy1}. Similarly one can prove \eqref{eq:crxy2}.
\end{proof}

\subsection{Proof of Theorem \ref{Main} (2)}
\begin{thm}
\label{TE:C}
For all $f\in D(T^{-1/2}), \psi\in D(\SQ(T)^{1/2})$ and $t\in\R$, 
\begin{align*}
e^{it H(\la)}C(f)e^{-it H(\la)}\psi=&C(e^{it\xi}f)\psi,\\
e^{it H(\la)}C(f)^*e^{-it H(\la)}\psi=&C(e^{it\xi}f)^*\psi.
\end{align*}
\end{thm}

\begin{proof}
These are proved in the same way as in the proof of Theorem \ref{TE} and Theorem \ref{CR:C}.
\end{proof}

\begin{lemma}
Let $\Om:=\V^{-1}\Omz$ where $\V$ is the unitary operator in Lemma \ref{UC(f)U}. Then:
\begin{enumerate}
\item
There is an eigenvalue $\tilde{\Eg}$ of $H(\la)$ and 
$\Om$ is an eigenvector of $H(\la)$ with eigenvalue $\tilde{\Eg}$.
\item
The following equation holds:
\begin{equation*}
\V H(\la)\V^{-1}=\SQ(\xi)+\tilde{\Eg}.
\end{equation*}
\item
The constant $\tilde{\Eg}$ is given as follows:
\begin{equation}
\label{eq:egtilde}
\tilde{\Eg}=\Eg-\beta\|\Tm\Ub\|^2.
\end{equation}
\end{enumerate}
\end{lemma}

\begin{proof}
Parts (1) and (2) can be proved in the same way as in the proof of Theorem \ref{Main} (1). \\
(3)\ We have 
\begin{equation*}
\tilde{\Eg}=\frac{\la}{4}\|g\|^2-\Tr(\xi^{1/2}Y^*Y\xi^{1/2})
\end{equation*}
in the same way as in the proof of Lemma \ref{gse}. 
Then, by Lemma \ref{xi1/2}, we have 
\begin{align*}
\xi^{1/2}Y^*Y\xi^{1/2}
=\Omp T^{1/2}V^*VT^{1/2}\Omp^*+\Omp T^{1/2}V^*\beta^{1/2}\Tm P
+\beta^{1/2}P\Tm VT^{1/2}\Omp^*+\beta P\Tm\Tm P.
\end{align*}
We choose a CONS $\{e_n\}_{n=0}^{\infty}\subset D(T)$ satisfying $e_0=\Ub$. Then it is easy to see that $\{\Omp^*e_n\}_{n=1}^{\infty}$ is a CONS for $\msH$ by Lemma \ref{omom=i}. 
Hence we have 
\begin{align*}
\Tr(\xi^{1/2}Y^*Y\xi^{1/2})=&\sum_{n=1}^{\infty}
\inner<e_n,\Omp T^{1/2}V^*VT^{1/2}\Omp^*e_n>+\beta\|\Tm U_b\|^2\\
=&\Tr(T^{1/2}V^*VT^{1/2})+\beta\|\Tm U_b\|^2.
\end{align*}
Thus we obtain \eqref{eq:egtilde}.
\end{proof}

In particular, $H(\la)$ have eigenvectors as follows:
\begin{equation*}
\phi_n:=\V^{-1}A(\Ub)^{*n}\Om_0,\ H(\la)\phi_n=(n\beta+\tilde{\Eg})\phi_n\ ,n\in\N\cup\{0\}.
\end{equation*}
Hence the spectral properties of $H(\lambda)$ as stated in Theorem \ref{Main} (2) follow.

\section{Analysis in the case $\la<\lacz$.}
In this section, we show that $H(\lambda)$ is unbounded from above and below.
\begin{thm}
Let $g\in D(T^{-1/2})$. 
Then $H(\la)$ is unbounded above for any $\la\in\R$.
If $\la<\lacz$, then $H(\la)$ is unbounded below too. 
\end{thm}

\begin{proof}
For any $f\in D(T)\backslash\{0\}$, we set $\psi_n:=a_n A(f)^{*n}\Om_0,\ a_n\in\C\sla\{0\},\ n\in\N\cup\{0\}$. Then we have the following equations: 
\begin{equation*}
\begin{array}{ll}
\SQ(T)\psi_n=n\dfrac{a_n}{a_{n-1}}A(Tf)^*\psi_{n-1},&A(g)\psi_n=n\inner<g,f>\dfrac{a_n}{a_{n-1}}\psi_{n-1}, \\
\|\psi_n\|^2=|a_n|^2 n!\|f\|^{2n},&\|A(g)^*\psi_n\|^2=\|g\|^2\|\psi_n\|^2+\|A(g)\psi_n\|^2.
\end{array}
\end{equation*}
Then we have 
\begin{equation*}
\inner<\psi_n,H(\la)\psi_n>
=\|\psi_n\|^2\left(\frac{\la}{4}\|g\|^2+n\frac{2\|T^{1/2}f\|^2+\la|\inner<g,f>|^2}{2\|f\|^2}\right).
\end{equation*}
We take $f$ such that $\inner<g,f>=0$. 
Then we have $\inner<\psi_n,H(\la)\psi_n>/\|\psi_n\|^2\ra\infty$ as $n\ra\infty$ for any $\la\in\R$. 
Thus $H(\la)$ is unbounded above for any $\la\in\R$. 

Let $\phi_N:=\sum_{n=0}^N\psi_n, N=0,1,2,\ldots$. 
Then we have $\|\phi_N\|^2=\sum_{n=0}^{N}\|\psi_n\|^2$ and 
\begin{align*}
\inner<\phi_N,H(\la)\phi_N>
=&\sum_{n=2}^N \|\psi_n\|^2
\left(\frac{\la\|g\|^2}{4}+n\frac{2\|T^{1/2}f\|^2+\la|\inner<g,f>|^2}{2\|f\|^2}
+\frac{\la}{2}\Re\frac{{a_{n-2}}^*}{{a_n}^*}\frac{\inner<g,f>^2}{\|f\|^4}\right)\\
&+\|\psi_1\|^2\left(\frac{\la\|g\|^2}{4}+\frac{\|T^{1/2}f\|^2}{\|f\|^2}
+\frac{\la|\inner<g,f>|^2}{2\|f\|^2}\right)+\frac{\la\|\psi_0\|^2\|g\|^2}{4}.
\end{align*}
Let $a_0:=1,\ a_n:=n^{-3/4}n!^{-1/2}, n\in\N$ and, for any $0<\delta,\ 0<\vep<1$, 
\begin{align*}
&f=f_{\delta}:=\frac{T^{-1}E((\delta,\infty))g}{\|T^{-1}E((\delta,\infty))g\|},\\ 
&c_{\la}(\vep,\delta):=\|T^{1/2}f_{\delta}\|^2\left\{1+\frac{\la}{2}(2-\vep)\|T^{-1/2}E((\delta,\infty))g\|^2\right\}.
\end{align*}
Then $\sum_{n=0}^{\infty}\|\psi_n\|^2$ converges and, for any $N\in\N$,  
\begin{equation}
\label{eq:unbddhami}
\inner<\phi_N,H(\la)\phi_N>
=\sum_{n=2}^N \|\psi_n\|^2n c_{\la}(\vep,\delta)
+\frac{\la}{2}\sum_{n=2}^N \|\psi_n\|^2
\left(\frac{a_{n-2}}{a_n}-n(1-\vep)\right)\inner<g,f_{\delta}>^2+C_N,
\end{equation}
where 
\begin{equation*}
C_N:=\frac{\la\|g\|^2}{4}\sum_{n=0}^N \|\psi_n\|^2
+\|\psi_1\|^2\left(\|T^{1/2}f_{\delta}\|^2+\frac{\la}{2}\inner<g,f_{\delta}>^2\right).
\end{equation*}
For all $0<\delta, 0<\vep<1$, we have 
\begin{equation}
\label{eq:ineq}
-\frac{2}{\|T^{-1/2}E((\delta,\infty))g\|^2(2-\vep)}<\lacz.
\end{equation}
The left hand side of \eqref{eq:ineq} tends to $\lacz$ as $\vep,\delta\dra0$. 
Since $\la<\lacz$, we can take a pair $(\vep,\delta)$ satisfying $c_{\la}(\vep,\delta)<0$. 
We fix such a pair. 
There is a $n_0\in\N$ such that $a_{n-2}/a_n-n(1-\vep)>0$ for all $n\geq n_0$. 
Hence we can see that $\inner<\phi_N,H(\la)\phi_N>/\|\phi_N\|^2$ tends to $-\infty$ as $N\ra\infty$, because the first term on the right hand side of \eqref{eq:unbddhami} tends to $-\infty$ as $N\ra\infty$.
\end{proof}

\section{Generalization of the $\phi^2$-model}
In this section we consider $H(\eta,\lambda)$ defined in Subsection \ref{Hamiltonian}. 
\begin{ass}
\label{ass:sa}
We need the following assumptions:
\begin{enumerate}
\item
$f\in D(T^{1/2})$ and $g\in D(T^{-1/2})\cap D(T^{1/2})$,
\item
$f\in D(T^{-1})$ and $\Re\inner<T^{-1}f,g>=0$,
\item
$f,g\in D(T^{-1})$ and $\Re\inner<T^{-1}f,g>\neq0$.
\end{enumerate}
\end{ass}

We can prove a slight generalization of Theorem \ref{Main}.
\begin{thm}
\label{Main2}
Let $\msH$ be separable. Then the following $(1)$-$(5)$ hold:
\begin{enumerate}
\item
Suppose that Assumption \ref{ass} and Assumption \ref{ass:sa} $(2)$ or $(3)$ hold. 
Let $\la>\lac$. Then there is a unitary operator $\U$ on $\fb(\msH)$ 
such that for all $\eta\in\R$,
\begin{equation*}
\U H(\eta,\la)\U^{-1}=\SQ(T)+\Eg+E_{f,g},
\end{equation*}
where the constant $E_{f,g}\in\R$ is defined by 
\begin{equation*}
E_{f,g}=-\dfrac{\eta^2}{2}\|T^{-1/2}f\|^2
+\dfrac{(\Re\inner<T^{-1}f,g>)^2\eta^2\la}{2(1+\la\|T^{-1/2}g\|^2)}.
\end{equation*}
\item
Suppose that Assumption \ref{ass} and Assumption \ref{ass:sa} $(2)$ or $(3)$ hold. 
Let $E_0>0$ and $\la>\lac$. 
Then there are a unitary operator $\V$ on $\fb(\msH)$ and a non-negative, injective self-adjoint operator $\xi$ on $\msH$such that, for all $\eta\in\R$,
\begin{equation*}
\V H(\eta,\la)\V^{-1}=\SQ(\xi)+\Eg-\Eb+E_{f,g}.
\end{equation*}
\item
Let $T$ be a non-negative, injective self-adjoint operator and suppose that $f$ and $g$ satisfy Assumption \ref{ass:sa} $(1)$ and $(2)$. Then there is a unitary operator $\W$ on $\fb(\msH)$ such that, for all $\eta\in\R$,
\begin{equation*}
\W\overline{H(\eta,\lacz)}\W^{-1}=\overline{H(\lacz)}-\frac{\eta^2}{2}\|T^{-1/2}f\|^2.
\end{equation*}
\item
Let $T$ be a non-negative, injective self-adjoint operator and suppose that $f$ and $g$ satisfy Assumption \ref{ass:sa} $(1)$ and $(3)$. Then, for all $\eta\in\R\sla\{0\}$,
\begin{equation*}
\sigma(\overline{H(\eta,\lacz)})=\R,\quad \sigp(\overline{H(\eta,\lacz)})=\emps.
\end{equation*}
\item
Let $T$ be a non-negative, injective self-adjoint operator and suppose that $f$ and $g$ satisfy Assumption \ref{ass:sa} $(1)$. Moreover, suppose that Assumption \ref{ass:sa} $(2)$ or $(3)$ holds. 
Let $\la<\lacz$. Then, for all $\eta\in\R$, $\overline{H(\eta,\la)}$ is unbounded above and  below. 
\end{enumerate}
\end{thm}

It is easy to see that Theorem \ref{Main2} is proved by 
the following lemma and facts in Theorem \ref{Main}.
\begin{lemma}
Let $T$ be a non-negative, injective self-adjont operator, $f\in D(T^{-1})$ and $g\in D(T^{-1/2})\cap D(T)$.
\begin{enumerate}
\item
Let $\Re\inner<T^{-1}f,g>=0$. Then there is a unitary operator $\U_1$ on $\fb(\msH)$ such that for all $\eta,\la\in\R$,
\begin{equation}
\U_1\overline{H(\eta,\la)}\U_1^{-1}=\overline{H(\la)}-\frac{\eta^2}{2}\|T^{-1/2}f\|^2. \label{eq:841}
\end{equation}
\item
Let $\Re\inner<T^{-1}f,g>\neq0$ and $g\in D(T^{-1})$.
\begin{enumerate}
\item
If $\la\neq\lacz$, then there is a unitary operator $\U_2$ on $\fb(\msH)$ such that for all $\eta\in\R$,
\begin{equation*}
\U_2\overline{H(\eta,\la)}\U_2^{-1}=\overline{H(\la)}+E_{f,g}.
\end{equation*}
\item
If $\la=\lacz$, then for all $\eta\in\R\sla\{0\}$,
\begin{equation}
\sigma(\overline{H(\eta,\lacz)})=\R,\quad \sigp(\overline{H(\eta,\lacz)})=\emptyset. \label{eq:843}
\end{equation}
\end{enumerate}
\end{enumerate}
\end{lemma}

\begin{proof}
	Let $\mathbb{U}_1:=e^{-i\Phi _{\rm s}(i\eta T^{-1}f)}$ for all $\eta \in \mathbb{R}$. Then, by direct calculation, we obtain 
	\begin{align}
	\mathbb{U}_1H(\eta ,\lambda )\mathbb{U}_1^{-1}=H(\lambda )-\frac{\eta ^2}{2}\|T^{-1/2}f \|^2-\lambda\eta \kappa \Phi _{\rm s}(g)+\frac{\lambda}{2}\eta ^2\kappa ^2 \label{eq:84p1}
	\end{align}
	on $\fbfin(D(T))$ 
      for all $\eta ,\lambda \in \R$, where $\kappa :=\Re \langle T^{-1}f,g\rangle $.		
	In the case of (1), we have (\ref{eq:841}) by $\kappa =0$ and a limit argument.	
	Next, we prove (2). We assume that $g\in D(T^{-1})$ and $\Re \langle T^{-1}f,g\rangle \neq0$. 
Let $\mathbb{V}_1:=e^{i\Phi _{\rm s}(i\alpha T^{-1}g)}$ for any $\alpha \in \mathbb{R}$ and  define a unitary operator $\mathbb{U}_2:=\mathbb{V}_1\mathbb{U}_1$. Then
	\begin{align*}
	\mathbb{U}_2H(\eta, \lambda)\mathbb{U}_2^{-1}=&H(\lambda)+\Bigl( \alpha +\lambda \alpha \|T^{-1/2}g\|^2 -\lambda \eta \kappa   \Bigr)\Phi _{\rm s}(g)\\
	&-\frac{\eta ^2}{2}\| T^{-1/2}f\|^2 +\frac{\lambda }{2}\eta ^2\kappa ^2+\frac{\alpha }{2}\| T^{-1/2}g \|^2\Bigl( \alpha +\lambda \alpha \| T^{-1/2}g\|^2-2\lambda \eta \kappa \Bigr)
	\end{align*}
	 on $\fbfin(D(T))$ in the same way as in (\ref{eq:84p1}). For $\lambda \not= \lambda _{{\rm c},0}$, let $\alpha =\lambda \eta \kappa (1+\lambda \|T^{-1/2}g\|^2)^{-1}$. Then we obtain 
	 \begin{align}
	 \mathbb{U}_2\overline{H(\eta, \lambda)}\mathbb{U}_2^{-1}=\overline{H(\lambda)}
-\frac{\eta ^2}{2}\| T^{-1/2}f\|^2 +\frac{\lambda \eta ^2\kappa ^2}{2(1+\lambda \|T^{-1/2}g\|^2)}\label{eq:84p2}
	 \end{align}
	 by a limit argument. If $\lambda =\lambda _{{\rm c},0}$, then, for all $\eta,\alpha\in\R$,  we have 
	 \begin{align*}
	 \mathbb{U}_2\overline{H(\eta, \lambda _{{\rm c},0})}\mathbb{U}_2^{-1}=\overline{H_g(-\kappa \eta \lambda _{{\rm c},0},\lambda )}-\frac{\eta^2}{2}\|T^{-1/2}f\|^2+\frac{\lambda _{{\rm c},0}\eta ^2\kappa ^2}{2}+\kappa \eta \alpha
	 \end{align*}
	in the same way as in \eqref{eq:84p2}, where $H_g(\nu, \lambda _{{\rm c},0}):=H(\lambda _{{\rm c},0})+\nu \Phi _{\rm s}(g)$ for all $\nu \in \mathbb{R}.$ It is easy to see that $\sigma (\overline{H_g(\nu, \lambda _{{\rm c},0})})=\mathbb{R}$ and $\sigma_{\rm p} (\overline{H_g(\nu, \lambda _{{\rm c},0})})=\emptyset$ for all $\nu \in \mathbb{R}\backslash\{0\}$, because $\mathbb{V}_1\overline{H_g(\nu ,\lambda _{{\rm c},0})}\mathbb{V}_1^{-1}=\overline{H_g(\nu ,\lambda _{{\rm c},0})}+\nu\alpha \|T^{-1/2}g\|^2$ and $\alpha \in \mathbb{R}$ is arbitrary. Hence we have (\ref{eq:843}).
\end{proof}

\begin{rem}
If $\msH$ is separable, then the condition $g\in D(T^{-1/2})\cap D(T)$ in the above lemma is weakened to the condition $g\in D(T^{-1/2})\cap D(T^{1/2})$.
\end{rem}

\section{Appendix}
In this section, we recall some known facts in Fock space theory. 
Let $T$ be a non-negative, injective self-adjoint operator on $\msH$.
\begin{lemma}{\rm \cite[Theorem 5.16.]{AR2}}\\
Let $f\in D(T^{-1/2})$ and $\psi\in D(\SQ(T)^{1/2})$. Then $\psi\in D(A(f))\cap D(A(f)^*)$ and the following inequalities hold$:$
\begin{align}
\|A(f)\psi\|\leq&\|T^{-1/2}f\|\|\SQ(T)^{1/2}\psi\|,\label{eq:ani}\\
\|A(f)^*\psi\|^2\leq&\|T^{-1/2}f\|^2\|\SQ(T)^{1/2}\psi\|^2+\|f\|^2\|\psi\|^2.\label{eq:cre}
\end{align}
\end{lemma}

\begin{lemma}{\rm \cite[Proposition 5.10.]{AR2}}
For any $f\in D(T)$, the following commutation relations on $\fbfin(D(T)):$
\begin{equation}
\label{eq:crsq}
[\SQ(T),A(f)]=-A(Tf),\quad [\SQ(T),A(f)^*]=A(Tf)^*.
\end{equation}
\end{lemma}

\begin{lemma}{\rm \cite[Lemma 5.21.]{AR2}}
\label{te}
For any $t\in\R, f\in\msH$, the following equation holds:
\begin{equation*}
e^{it\SQ(T)}A(f)^{\sharp}e^{-it\SQ(T)}=A(e^{itT}f)^{\sharp}.
\end{equation*}
\end{lemma}

\begin{lemma}{\rm \cite[Theorem 5.21.]{AR2}}
\label{conv}
Assume that $\msH$ be separable. Let $T$ be a non-negative, injective self-adjoint operator and $\{e_n\}_{n=1}^{\infty}\subset D(T^{1/2})$ be a CONS of $\msH$. Then, for any $\psi\in D(\SQ(T)^{1/2})$, $\sum_{n=1}^{\infty}\|A(T^{1/2}e_n)\psi\|^2$ converges and 
following equation holds$:$
\begin{equation*}
\sum_{n=1}^{\infty}\|A(T^{1/2}e_n)\psi\|^2=\|\SQ(T)^{1/2}\psi\|^2.
\end{equation*}
\end{lemma}


\begin{thebibliography}{20}


\bibitem{AR81}Arai A., On a model of a harmonic oscillator coupled to a quantized, massless, scalar field. I, J. Math. Phys. {\bf 22} (1981), 2539-2548.

\bibitem{AR89}Arai A., Spectral analysis of a quantum harmonic oscillator coupled to infinitely many scalar bosons, J. Math. Anal. Appl. {\bf 140} (1989), 270-288.

\bibitem{AR1}Arai A., Fock Spaces and Quantum Fields, Nippon Hyoronsya, 2000, in Japanese.

\bibitem{ARUn}Arai A., A note on mathematical analysis of a pair-interaction model in quantum field theory, Unpublished, 2017. 

\bibitem{AR2}Arai A., Analysis on Fock Spaces and Mathematical Theory of Quantum Fields, World Scientific, Singapore, 2018.

\bibitem{Be1}Berezin F. A., The Method of Second Quantization, Academic Press, New York, 1966.

\bibitem{De17}Derezi$\acute{\mathrm n}$ski J., Bosonic quadratic Hamiltonians, J. Math. Phys. {\bf 58} (2017), 121101-45.

\bibitem{HT1}Henley E. M. and Thirring W., Elementary Quantum Field Theory, McGraw-Hill Book Company, New York, 1962.

\bibitem{Hi12}Hiroshima F., Sasaki I, Spohn H. and Suzuki A., Enhanced Binding in Quantum Field Theory, COE Lecture Note Vol.38, IMI, Kyushu University, 2012. 

\bibitem{MS05}Miyao T. and Sasaki I., Stability of discrete ground state, Hokkaido Math. J. {\bf 34} (2005), 689-717.

\bibitem{NNS16}Nam P. T., Napi$\acute{\mathrm o}$rkowski M. and Solovej J. P.,  Diagonalization of bosonic quadractic Hamiltonians by Bogoliubov transformations, J. Func. Anal, {\bf 270} (2016), 4340-4368.

\bibitem{BS79}Reed M. and Simon B., Methods of Modern Mathematical Physics I: Functional Analysis, Academic Press, New York, 1979.

\bibitem{Ru78}Ruijsenaars S. N. M., On Bogoliubov transformations. II. The general case, Ann. Phys. {\bf 116} (1978), 105-134.

\bibitem{Sh62}Shale D., Linear symmetries of free boson fields, Trans. Amer. Math. Soc. {\bf  103} (1962), 149-167.

\bibitem{TR15}Teranishi N., Self-adjointness of the generalized spin-boson Hamiltonian with a quadratic boson interaction, Hokkaido Math. J. {\bf 44}, no. 3 (2015), 409-423.

\bibitem{Tit1}Titchmarsh E. C., Introduction to the Theory of Fourier Integrals, Oxford U.P., Amen House, London, 1937.
\end{thebibliography}
\end{document}